\documentclass[onecolumn,american,english,final,table,twocolumn]{IEEEtran}
\usepackage[T1]{fontenc}
\usepackage[latin9]{inputenc}
\usepackage{array}
\usepackage{float}
\usepackage{booktabs}
\usepackage{mathrsfs}
\usepackage{bm}
\usepackage{multirow}
\usepackage{amsmath}
\usepackage{amsthm}
\usepackage{amssymb}
\usepackage{graphicx}

\makeatletter

\providecommand{\tabularnewline}{\\}
\floatstyle{ruled}
\newfloat{algorithm}{tbp}{loa}
\providecommand{\algorithmname}{Algorithm}
\floatname{algorithm}{\protect\algorithmname}

\theoremstyle{plain}
\newtheorem{prop}{\protect\propositionname}
\theoremstyle{plain}
\newtheorem{fact}{\protect\factname}

\usepackage{amsmath}
\usepackage{amssymb}

\usepackage{algorithm}
\usepackage{algorithmic}
\usepackage{array}
\usepackage{booktabs}

\usepackage{graphicx,psfrag,cite,subfigure}

\author{
\thanks{This is a preprint version. The paper has been published at IEEE Transactions on Signal Processing, vol. 73, pp. 2833-2847, 2025, doi: 10.1109/TSP.2025.3584229. Junting Chen is the corresponding author. The work was supported in part by the Basic Research Project No. HZQB-KCZYZ-2021067 of Hetao Shenzhen-HK S\&T Cooperation Zone, by the National Science Foundation of China under Grant No. 62171398 and No. 62293482, by the Guangdong Basic and Applied Basic Research Foundation 2024A1515011206, by the Shenzhen Science and Technology Program under Grant No. JCYJ20220530143804010, No. KQTD20200909114730003, and No. KJZD20230923115104009, by Guangdong Research Projects No. 2019QN01X895, No. 2017ZT07X152, and No. 2019CX01X104, by the Shenzhen Outstanding Talents Training Fund 202002, by the Guangdong Provincial Key Laboratory of Future Networks of Intelligence (Grant No. 2022B1212010001), by the National Key R\&D Program of China with grant No. 2018YFB1800800, and by the Key Area R\&D Program of Guangdong Province with grant No. 2018B030338001.}

\IEEEauthorblockN{Yuanshuai~Zheng and Junting~Chen}

 \IEEEauthorblockA{School of Science and Engineering,
Shenzhen Future Network of Intelligence Institute (FNii-Shenzhen), and\\
Guangdong Provincial Key Laboratory of Future Networks of Intelligence\\
The Chinese University of Hong Kong, Shenzhen, Guangdong 518172, P. R. China}

}


\usepackage[acronym]{glossaries}
\newcommand{\newac}{\newacronym}
\newcommand{\ac}{\gls}
\newcommand{\Ac}{\Gls}
\newcommand{\acpl}{\glspl}

\makeglossaries

\newac{speb}{SPEB}{square position error bound}
\newac[plural=EFIMs,firstplural=Fisher information matrices (EFIMs)]{efim}{EFIM}{Fisher information matrix}
\newac{ne}{NE}{Nash equilibrium}
\newac{mse}{MSE}{mean squared error}
\newac{toa}{TOA}{time-of-arrival}
\newac{snr}{SNR}{signal-to-noise ratio}
\newac{lan}{LAN}{local area network}
\newac{psd}{PSD}{positive semidefinite}
\newac{pd}{PD}{positive definite}
\newac{wrt}{w.r.t.}{with respect to}
\newac{lhs}{L.H.S.}{left hand side}
\newac{wp1}{w.p.1}{with probability 1}
\newac{kkt}{KKT}{Karush-Kuhn-Tucker}
\newac{wlog}{w.l.o.g.}{without loss of generality}
\newac{mle}{MLE}{maximum likelihood estimation}
\newac{gps}{GPS}{global positioning system}
\newac{rssi}{RSSI}{received signal strength indicator}
\newac{mimo}{MIMO}{multiple-input multiple-output}
\newac{csi}{CSI}{channel state information}
\newac{fdd}{FDD}{frequency division duplexing}
\newac{ms}{MS}{mobile station}
\newac{bs}{BS}{base station}
\newac{d2d}{D2D}{device-to-device}
\newac{slnr}{SLNR}{signal-to-interference-leakage-and-noise-ratio}
\newac{ula}{ULA}{uniform linear antenna array}
\newac{pas}{PAS}{power angular spectrum}
\newac{mmse}{MMSE}{minimum mean square error}
\newac{zf}{ZF}{zero-forcing}
\newac{rzf}{RZF}{regularized zero-forcing}
\newac{as}{AS}{angular spread}
\newac{aod}{AOD}{angle of departure}
\newac{iid}{i.i.d.}{independent and identically distributed} 
\newac{sinr}{SINR}{signal-to-interference-and-noise ratio}
\newac{tdd}{TDD}{time-division duplex}
\newac{rvq}{RVQ}{random vector quantization}
\newac{rhs}{R.H.S.}{right hand side}
\newac{mrc}{MRC}{maximum ratio combining}
\newac{cdf}{CDF}{cumulative distribution function}
\newac{a.s.}{a.s.}{almost surely}
\newac{los}{LOS}{line-of-sight}
\newac{jsdm}{JSDM}{joint spatial division and multiplexing}
\newac{map}{MAP}{maximum a posteriori}
\newac{klt}{KLT}{Karhunen-Lo\`eve Transform}
\newac{lbe}{LBE}{link bargaining equilibrium}
\newac{se}{SE}{Stackelberg equilibrium}
\newac{uav}{UAV}{unmanned aerial vehicle}
\newac{nlos}{NLOS}{non-line-of-sight}
\newac{pdf}{PDF}{probability density function}
\newac{em}{EM}{expectation-maximization}
\newac{knn}{KNN}{$k$-nearest neighbor}
\newac{svd}{SVD}{singular value decomposition}
\newac{nmf}{NMF}{non-negative matrix factorization}
\newac{umf}{UMF}{unimodality-constrained matrix factorization}
\newac{rmse}{RMSE}{rooted mean squared error}
\newac{olos}{OLOS}{obstructed line-of-sight}
\newac{mmw}{mmW}{millimeter wave}
\newac{ber}{BER}{bit error rate}
\newac{rss}{RSS}{received signal strength}
\newac{lp}{LP}{linear program}
\newac{ufw}{U-FW}{unimodal Frank-Wolfe}
\newac{utf}{UTF}{unimodality-constrained tensor factorization}
\newac{fw}{FW}{Frank-Wolfe}
\newac{iot}{IoT}{Internet-of-Things}
\newac{mae}{MAE}{mean absolute error}
\newac{crb}{CRB}{Cram\'er-Rao bound}
\newac{aoa}{AoA}{angle of arrival}
\newac{wcl}{WCL}{weighted centroid localization}
\newac{thz}{THz}{Terahertz}
\newac{isac}{ISAC}{integrated sensing and communication}
\newac{gnss}{GNSS}{global navigation satellite system}
\newac{3d}{3D}{three-dimensional}
\newac{2d}{2D}{two-dimensional}
\newac{pso}{PSO}{particle swarm optimization}
\newac{dp}{DP}{dynamic programming}
\newac{ao}{AO}{alternating optimization}
\newac{sca}{SCA}{successive convex approximation}
\newac{hmm}{HMM}{Hidden Markov Model}
\newac{pca}{PCA}{Principal Component Analysis}
\newac{dft}{DFT}{discrete Fourier transform}
\newac{ar}{AR}{autoregressive}
\newac{kf}{KF}{Kalman filter}
\newac{lstm}{LSTM}{long-short term memory}
\newac{nmse}{NMSE}{normalized mean squared error}
\newac{siso}{SISO}{single-input single-output}
\newac{ls}{LS}{least squares}
\newac{rsrp}{RSRP}{reference signal received power}
\newac{cs}{CS}{compressive sensing}
\newac{rf}{RF}{radio-frequency}
\newac{cam}{CAM}{channel angle map}
\newac{skf}{SKF}{switching Kalman filter}
\newac{evd}{EVD}{eigenvalue decomposition}
\newac{bim}{BIM}{beam index map}
\newac{wwcl}{WWCL}{weighted centroid correction localization}
\newac{uwb}{UWB}{ultra-wideband}
\newac{miso}{MISO}{Multiple-Input Single-Output}

\setkeys{Gin}{width=1.0\columnwidth}

\IEEEoverridecommandlockouts

\makeatother

\usepackage{babel}
\addto\captionsamerican{\renewcommand{\algorithmname}{Algorithm}}
\addto\captionsamerican{\renewcommand{\factname}{Fact}}
\addto\captionsamerican{\renewcommand{\propositionname}{Proposition}}
\addto\captionsenglish{\renewcommand{\factname}{Fact}}
\addto\captionsenglish{\renewcommand{\propositionname}{Proposition}}
\providecommand{\factname}{Fact}
\providecommand{\propositionname}{Proposition}

\begin{document}
\title{A Radio Map Approach for Reduced Pilot CSI Tracking in Massive MIMO
Networks}

\maketitle
\selectlanguage{american}%
%
%



\ifdefined\SINGLECOLUMN
	\setkeys{Gin}{width=0.5\columnwidth}
	\newcommand{\figfontsize}{\footnotesize} 
	\newcommand{\labelfontsize}{0.7}
	\newcommand{\legendfontsize}{0.57}
	\newcommand{\ticklabelfontsize}{0.7}
	\newcommand{\numberfontsize}{0.45}
\else
	\setkeys{Gin}{width=1.0\columnwidth}
	\newcommand{\figfontsize}{\normalsize} 
	\newcommand{\labelfontsize}{0.7}
	\newcommand{\legendfontsize}{0.65}
	\newcommand{\ticklabelfontsize}{0.7}
	\newcommand{\numberfontsize}{0.5}
\fi
\newac{ode}{ODE}{ordinary differential equation}

\selectlanguage{english}%
\begin{abstract}

Massive \ac{mimo} systems offer significant potential to enhance
wireless communication performance, yet accurate and timely \ac{csi}
acquisition remains a key challenge. Existing works on \ac{csi} estimation
and radio map applications typically rely on stationary \ac{csi}
statistics and accurate location labels. However, the \ac{csi} process
can be discontinuous due to user mobility and environmental variations,
and inaccurate location data can degrade the performance. By contrast,
this paper studies radio-map-embedded \ac{csi} tracking and radio
map construction without the assumptions of stationary \ac{csi} statistics
and precise location labels. Using radio maps as the prior information,
this paper develops a radio-map-embedded \ac{skf} framework that
jointly tracks the location and the CSI with adaptive beamforming
for sparse CSI observations under reduced pilots. For radio map construction
without precise location labels, the location sequence and the channel
covariance matrices are jointly estimated based on a \ac{hmm}. An
unbiased estimator on the channel covariance matrix is found. Numerical
results on ray-traced \ac{mimo} channel datasets demonstrate that
using $1$ pilot observation in every $10$ milliseconds, an average
of over $97\%$ of capacity over that of perfect \ac{csi} can be
achieved, while a conventional \ac{kf} can only achieve $76\%$.
Furthermore, the proposed radio-map-embedded \ac{csi} model can reduce
the localization error from $30$ meters from the prior to $6$ meters
for radio map construction.
\end{abstract}


\glsresetall

\section{Introduction}

\label{sec:intro}

\selectlanguage{american}%
\IEEEPARstart{M}{assive}\foreignlanguage{english}{ \ac{mimo} enables
spatial multiplexing, provides high beamforming gain, and supports
concurrent multi-user data transmission. To realize the full potential
of massive \ac{mimo}, it is essential to acquire the up-to-date \ac{csi},
which is time-varying due to the mobility of the users and the variation
of the environment. Timely \ac{csi} acquisition is challenging in
massive \ac{mimo} systems because the dimension of the channel vector
can be very high \cite{GuoLau:J24,HeSheXuEld:J24}.}

\selectlanguage{english}%

Massive \ac{mimo} channel estimation and tracking generally require
some prior information about the channel characteristics. If the channel
covariance matrix is available, \ac{mmse}, or more generally, Bayesian
approaches can be employed to constrain solutions within the subspace
spanned by the covariance matrix \cite{NguSwiNgu:J24,MehSriAsiJag:J24}.
If the channel is believed to be sparse, compressive sensing techniques
can be applied to obtain sparse solutions \cite{KeGaoWuGao:J20,WanLiuLiuZha:J24}.
Furthermore, if the channel is temporally and spatially correlated,
\ac{kf} approaches and \ac{ar} models can be leveraged to track
the \ac{csi} process with reduced pilot overhead per time slot \cite{BroWanDas:J15,VinJunHam:C24}.
Recent advancements also explore \ac{lstm}-based methods \cite{ZhaZhaGaoYan:J24,PenZhaCheYan:C20}
and some other deep learning methods, such as attention networks \cite{ZhoYanMaGao:J24}
and generative models \cite{MirShaAdvGop:J22}.

However, if the prior information is inaccurate or varies over time,
the existing methods \cite{KeGaoWuGao:J20,VinJunHam:C24,WanLiuLiuZha:J24,BroWanDas:J15,ZhaZhaGaoYan:J24,PenZhaCheYan:C20,ZhoYanMaGao:J24,MirShaAdvGop:J22}
may perform poorly. For example, the channel covariance matrices may
differ across locations while channel dynamics follows different models
in \ac{los} scenarios and \ac{nlos} scenarios. Therefore, tracking
algorithms like \ac{kf} may lose track of the channel when the mobile
enters an \ac{nlos} region from an \ac{los} region.

Radio maps offer a promising solution to provide prior information
on channels \cite{ZenCheXuWu:J23,ZheJiaChe:C24,WuZenJinZha:J23}.
Existing works on radio-map-assisted \ac{csi} estimation typically
require location information from additional devices or via sensing.
Some location-based approaches \cite{ChuKIM:J24,WuZenJinZha:J23}
integrate user position information, obtained via \ac{gps} or \ac{uwb}
techniques, to enhance \ac{csi} prediction by exploiting spatial
correlations between locations and channels. Similarly, \ac{isac}
techniques \cite{YuHuLiuPen:J22,DuCheJinLi:J24} first sense user
positions and then utilize this information to assist in \ac{csi}
estimation. The work \cite{ZheJiaChe:C24} attempted radio-map-assisted
CSI estimation without explicit location information in the estimation
phase, where the radio map is a mapping from a location to a channel
realization. However, such a radio map is challenging to construct
without a massive amount of location-labeled data. The \ac{bim} proposed
in \cite{DaiWuDonLi:J24} requires exhaustive beam sweeping at each
location grid during the training process, and the beam alignment
is achieved based on the accurate user position and the \ac{bim}.
The work \cite{WuZen:C23} focuses on constructing a \ac{cam}, requiring
accurate information of the \ac{aod} that heavily depends on user
positions and surrounding environments. Although the \ac{cam}-based
method \cite{WuZen:C23} can estimate both the user position and the
\ac{csi} after deployment, their reliance on precise training data
of \acpl{aod} and user positions remains a significant limitation.
Some recent attempts on channel charting developed a \ac{csi} tracking
strategy without explicit user positions using deep learning techniques
\cite{StaYamFeiEsk:J24,FerGuiStuTir:T23}, but the inherent physical
mechanism is still not clear.

In this paper, we investigate the following two fundamental problems:
\begin{itemize}
\item {\em Can radio map help reduce pilot observations for CSI tracking?}
In particular, we focus on the case where there is {\em no} precise
location information available as an index for the \ac{csi} in the
radio map. Instead, location also needs to be recovered from the pilot
observations as an intermediate step for the \ac{csi} tracking. Thus,
we aim at a joint recovery of the location and \ac{csi} using radio
maps as the prior information.
\item {\em Can we reconstruct and update a radio map from pilot observation sequences without precise location labels?}
To construct a radio map, we need to recover precise location information
from pilot observation sequences with possibly some side information
such as the locations of \acpl{bs}. Thus, the question is, can a
radio-map-embedded \ac{csi} model help recover precise location information
despite possibly \ac{nlos}?
\end{itemize}

In contrast to the preliminary study \cite{ZheJiaChe:C24} which maps
a location to a specific channel realization, this work adopts a different
radio map methodology, where it defines a mapping from a location
to a channel distribution. As such, the new radio map definition is
robust to local scattering, body shadowing, and antenna pattern at
the user side. Consequently, it enables not only robust radio-map-embedded
CSI tracking, but also a blind construction of radio maps without
location labels. Towards this goal, this paper exploits the principle
that while the \ac{csi} process can be discontinuous in a high-dimensional
\ac{csi} space, the trajectory of the correspondingly mobile terminal
has to be continuous in the three-dimensional physical world. As such,
a sequence of very sparse \ac{csi} observations may suffice to discover
the mobile trajectory, which leads to accurate tracking of the high-dimensional
\ac{csi} with the help of radio maps. Following such a methodology,
we develop a radio-map-embedded \ac{skf} framework that jointly tracks
the location and the \ac{csi} with adaptive beamforming for sparse
\ac{csi} observations under reduced pilot observations, where radio
maps play an essential role in adjusting the operating regime of \ac{skf}.
In addition, we employ a Bayesian approach for a joint estimation
of the location and the channel covariance matrix from a massive amount
of offline unlabeled pilot observation sequences for the construction
of radio maps. The key theoretical contributions and numerical findings
are summarized as follows:
\begin{itemize}
\item We propose a radio-map-embedded \ac{skf} framework for \ac{csi}
tracking from a sequence of reduced pilot observations without location
labels. In addition, instead of using a randomized compressive sensing
approach, we develop an adaptive sensing matrix to minimize the differential
entropy of the \ac{csi}.
\item We develop a joint estimation on the location sequence and the channel
covariance matrices from unlabeled pilot observation sequences using
a \ac{hmm}. An unbiased estimator on the channel covariance matrix
is found with the estimation error analyzed.
\item We numerically demonstrate that, in a single-user $64$-antenna \ac{miso}
system, using $1$ pilot observation in every $10$ milliseconds,
an average of over $97\%$ of capacity over that of perfect \ac{csi}
can be achieved for a user moving at $36$~km/h at a $20$~dB \ac{snr}
under \ac{los}, while a conventional \ac{kf} can only achieve $76\%$.
In a pure \ac{nlos} regime, the proposed radio-map-assisted approach
can double the capacity from existing schemes under low to moderate
\ac{snr}. For radio map construction with unlabeled pilot observation
sequences, we demonstrate that the proposed estimation using the radio-map-embedded
\ac{csi} model reduces the localization error from $30$ meters from
the prior to $6$ meters, which is sufficient for radio map construction
in our framework.
\end{itemize}

\selectlanguage{american}%
The remainder of the paper is organized as follows: Section II introduces
the system model and formulates the problems of \ac{csi} tracking
and radio map construction. Section III presents a radio-map-embedded
\ac{skf} framework for fast \ac{csi} tracking. Section IV proposes
to iteratively estimate user positions and channel covariance for
radio map construction, along with the analysis of construction error.
Section V provides simulations and related discussions, while Section
VI concludes the paper.

\textit{Notation}: Vectors and matrices are denoted by bold lowercase
$\mathbf{x}$ and bold uppercase $\mathbf{X}$, respectively. Scalar
quantities are denoted by non-bold letters such as $x$. Collections
are denoted by script $\mathcal{X}$. Probability density and mass
functions are denoted by $p(\cdot)$ and $\mathbb{P}(\cdot)$, respectively.
The expectation and variance operators are represented by $\mathbb{E}(\cdot)$
and $\mathbb{V}(\cdot)$, respectively. The $L_{2}$ norm is denoted
by $\|\cdot\|_{2}$, and the Frobenius norm is denoted by $\|\cdot\|_{{\rm F}}$.
The conjugate transpose of a complex vector $\mathbf{x}$ or a complex
matrix $\mathbf{X}$ is denoted by ${\bf x}^{{\rm H}}$ or ${\bf X}^{{\rm H}}$,
respectively. The trace of a matrix ${\bf X}$ is denoted by ${\rm tr}({\bf X})$.
The pseudo-inverse of a matrix $\mathbf{X}$ is denoted by $\mathbf{X}^{\dagger}$.
The identity matrix is denoted by ${\bf I}$. Furthermore, we specify
several key variables that will be used throughout the paper: the
user position, channel state, sensing matrix, and observation at time
$t$ are denoted by $\mathbf{p}_{t}$, $\mathbf{h}_{t}$, $\mathbf{A}_{t}$,
and $\mathbf{y}_{t}$, respectively. Their corresponding sequences
up to time $t$ are denoted by $\mathcal{P}_{t}=(\mathbf{p}_{1},\mathbf{p}_{2},\dots,\mathbf{p}_{t})$,
$\mathcal{H}_{t}=({\bf h}_{1},{\bf h}_{2},\dots,{\bf h}_{t})$, $\mathcal{A}_{t}=(\mathbf{A}_{1},\mathbf{A}_{2},\dots,\mathbf{A}_{t})$,
and $\mathcal{Y}_{t}=(\mathbf{y}_{1},\mathbf{y}_{2},\dots,\mathbf{y}_{\text{\ensuremath{t}}})$,
with their realizations when $t=T$ denoted by $\mathcal{P}_{T}=(\mathbf{p}_{1},\mathbf{p}_{2},\dots,\mathbf{p}_{T})$,
$\mathcal{H}_{T}=({\bf h}_{1},{\bf h}_{2},\dots,{\bf h}_{T})$, $\mathcal{A}_{T}=(\mathbf{A}_{1},\mathbf{A}_{2},\dots,\mathbf{A}_{T})$,
and $\mathcal{Y}_{T}=(\mathbf{y}_{1},\mathbf{y}_{2},\dots,\mathbf{y}_{\text{\ensuremath{T}}})$,
respectively.
\selectlanguage{english}%

\section{System Model}

\label{subsec:System-model}

\subsection{Network Topology, Channel Model, and Radio Map Model}

\label{subsec:Network-Topology}

\begin{figure}
\begin{centering}
\includegraphics[width=1\columnwidth]{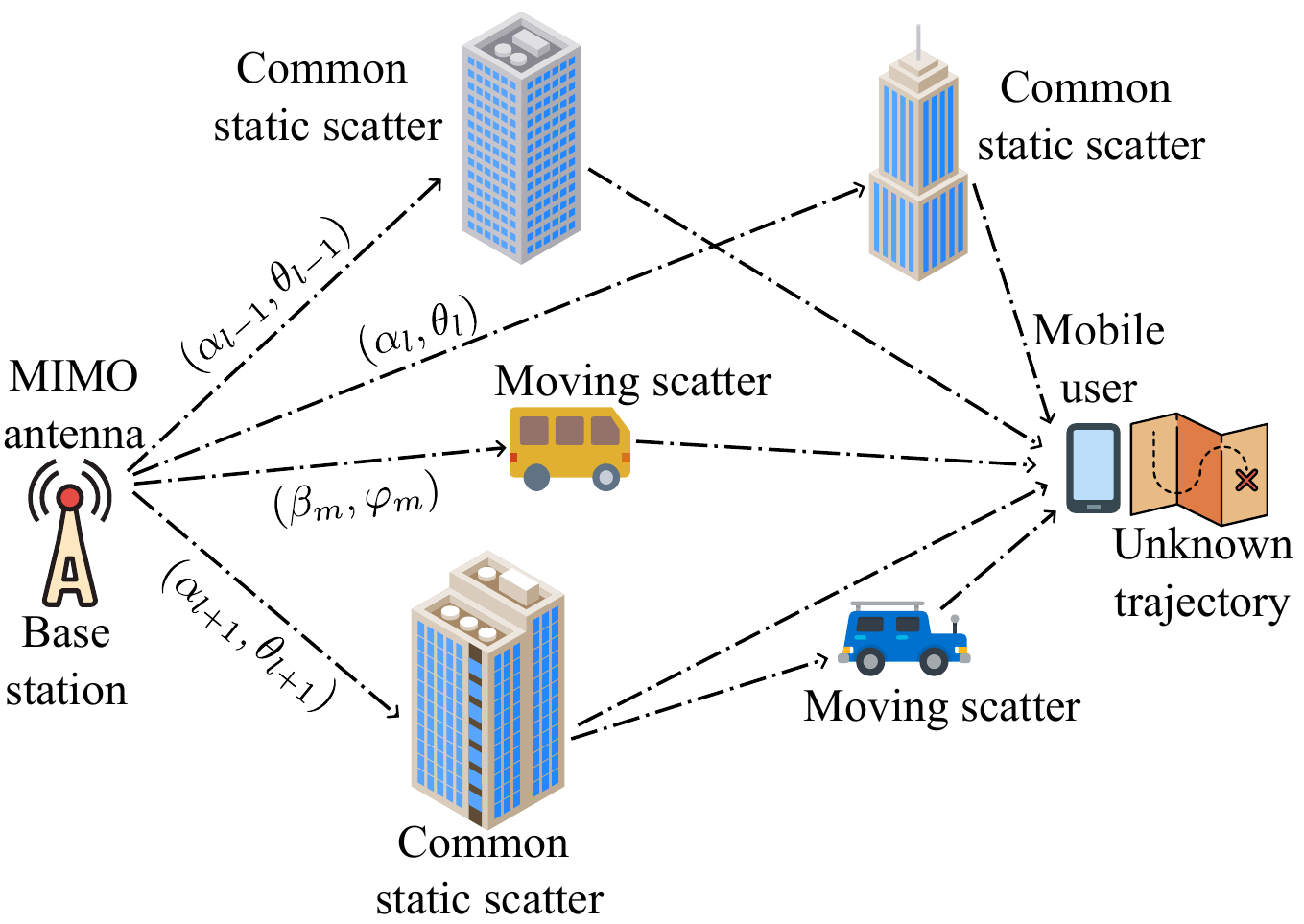}
\par\end{centering}
\caption{\label{fig:channel-model}Illustration of the channel in a quasi-static
environment, where the majority of the scatters are static and a small
portion of scatters are time-varying, and adjacent mobile locations
may share common static scatters.}
\end{figure}

Consider a \ac{mimo} network with $Q$ \acpl{bs} and one mobile
user, where each \ac{bs} is equipped with $N_{\mathrm{t}}$ antennas
and the user has a single antenna. To ease the elaboration, we simply
illustrate the case of $Q=1$ because, for $Q>1$, we can simply stack
variables and increase the dimension of the system states, where the
form of the formulas for the multi-\ac{bs} case is similar.

Suppose that the environment is quasi-static, where the majority of
the scatters are static and only a small portion of scatters are time-varying.
The static scatters may correspond to the ground and buildings, and
the time-varying scatters may correspond to the vehicles and pedestrians.
We are interested in the scenarios where the \ac{mimo} \acpl{bs}
are placed on the rooftop or tower, where there are few moving scatters
near the \acpl{bs}.

The propagation between the \ac{bs} and the user located at $\mathbf{p}$
can be modeled as signal aggregation from the set of propagation paths
$\mathcal{N}_{\text{s}}(\mathbf{p})$ via static scatters and from
the paths $\mathcal{N}_{\text{m}}(\mathbf{p})$ via moving scatters.
As illustrated in Fig.~\ref{fig:channel-model}, a geometric narrow-band
multipath \ac{mimo} channel can be modeled as \cite{ZhoPanRenPop:J22}
\begin{equation}
{\bf h}(\mathbf{p})=\sum_{l\in\mathcal{N}_{\text{s}}(\mathbf{p})}a_{l}(\mathbf{p})\bm{\alpha}(\theta_{l}(\mathbf{p}))+\sum_{m\in\mathcal{N}_{\text{m}}(\mathbf{p})}\beta_{m}\bm{\alpha}(\varphi_{m})\label{eq:channel-model}
\end{equation}
where the coefficients $a_{l}$ and $\beta_{m}$ denote the complex
path gains, $\theta_{l}(\mathbf{p})$ and $\varphi_{m}$ denote the
\ac{aod} at the \ac{bs}, and $\bm{\alpha}(\theta)$ denotes the
steering vector at $\theta$ according to the antenna geometry of
the \ac{bs}. The statistics of the coefficients $a_{l}(\mathbf{p})$
from the static scatters depend on the user location $\mathbf{p}$,
whereas the coefficients $\beta_{m}$ and the \ac{aod} $\varphi_{m}$
from the moving scatters are assumed to be independent of $\mathbf{p}$.

Based on (\ref{eq:channel-model}), the multipath channel can be written
as 
\begin{equation}
{\bf h}(\mathbf{p})\triangleq{\bf h}^{\text{s}}(\mathbf{p})+{\bf h}^{\epsilon}\label{eq:normalized-channel-model}
\end{equation}
where ${\bf h}^{\text{s}}(\mathbf{p})\triangleq\sum_{l\in\mathcal{N}_{\text{s}}(\mathbf{p})}a_{l}(\mathbf{p})\bm{\alpha}(\theta_{l}(\mathbf{p}))$
is the channel due to static scattering for user location ${\bf p}$,
and ${\bf h}^{\epsilon}\triangleq\sum_{m\in\mathcal{N}_{\text{m}}(\mathbf{p})}\beta_{m}\bm{\alpha}(\varphi_{m})$
is the channel due to the moving scatters. It is assumed that both
${\bf h}^{\text{s}}(\mathbf{p})$ and $\mathbf{h}^{\epsilon}$ have
zero mean and they are independent.

Consider a discrete location model, where the area of interest is
partitioned into equally spaced grids. The user is located at one
of the grid points $\mathbf{x}\in\mathcal{X}$, where $\mathcal{X}$
is a collection of all the grid locations. We define the radio map
as a collection of the discretized location $\mathbf{x}\in\mathcal{X}$
and the covariance matrix ${\bf C}({\bf x})$ pairs, {\em i.e.},
$\mathcal{M}=\{(\mathbf{x},{\bf C}({\bf x})):\mathbf{x}\in\mathcal{X}\}$,
where $\mathbf{C}(\mathbf{x})\triangleq\mathbb{E}\{\mathbf{h}\mathbf{h}^{\text{H}}\}$,
in which the expectation is taken over the distribution of the positions
$\mathbf{p}$ within the grid cell with a center location $\mathbf{x}\in\mathcal{X}$
and the small-scale fading of the channel. We assume that ${\bf h}^{\text{s}}(\mathbf{p})$
is spatially correlated for all the positions $\mathbf{p}$ in the
same grid cell and $\mathbb{E}\{{\bf h}^{\epsilon}({\bf h}^{\epsilon})^{\text{H}}\}=\sigma_{\text{h}}^{2}\mathbf{I}$,
where $\sigma_{\text{h}}^{2}$ is much smaller than the mean energy
of ${\bf h}^{\text{s}}(\mathbf{p})$. Hence, the distribution of the
channel $f_{\mathrm{h}}(\mathbf{h})$ given the location $\mathbf{x}$
is fully captured by the covariance matrix ${\bf C}({\bf x})$ and
the parameter $\sigma_{\mathrm{h}}^{2}$. With such a discretized
model, we now consider that $\mathbf{p}$ takes a discrete position
from $\mathcal{X}$ in the rest of the paper.
\begin{figure}
\begin{centering}
\includegraphics[width=1\columnwidth]{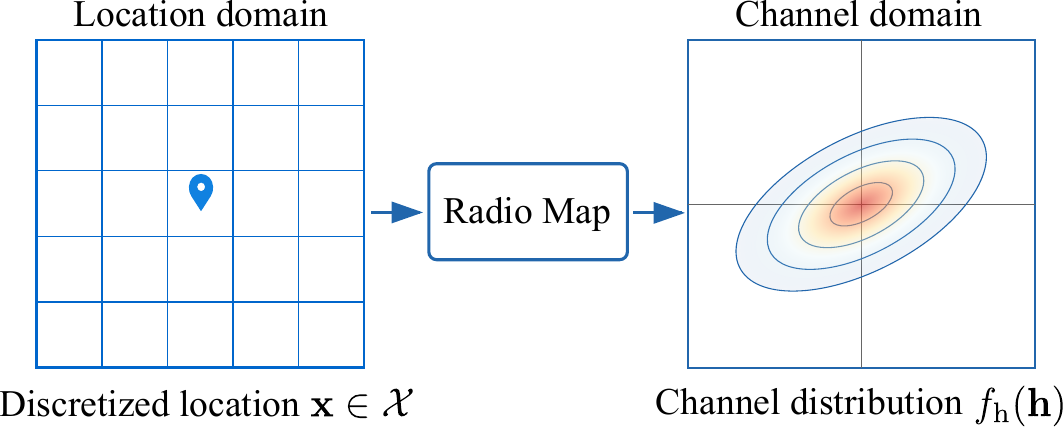}
\par\end{centering}
\caption{\label{fig:radio-map-illustration}Illustration of the radio map.
The area of interest is partitioned into grid cells. The radio map
is a mapping from a grid cell to a channel distribution $f_{\mathrm{h}}(\mathbf{h}|\mathbf{x}_{i})$,
where $\mathbf{x}_{i}$ is a discretized location representing the
$i$th grid cell.}
\end{figure}

\subsection{System Dynamics}

\label{subsec:System-Dynamics-and-the-CSI-Tracking-Problem}

Consider that the mobile user moves along a trajectory $\mathcal{P}_{t}=(\mathbf{p}_{1},\mathbf{p}_{2},\dots,\mathbf{p}_{t})$
that can be described by a Markov chain with known transition probability
$\mathbb{P}(\mathbf{p}_{\text{\ensuremath{t}}}|\mathbf{p}_{t-1})$.
Such prior information can be obtained empirically from prior measurements.
For example, one can assume a uniform prior $\mathbb{P}(\mathbf{p}_{\text{\ensuremath{t}}}=\mathbf{x}_{j}|\mathbf{p}_{t-1}={\bf x}_{i})$
for a reasonable range of speed of a vehicle. Alternatively, the transition
probability $\mathbb{P}({\bf p}_{t}={\bf x}_{i}|{\bf p}_{t-1}={\bf x}_{j})$
can be specified using a truncated Gauss-Markov model \cite{LiaHaa:C99}
over the coordinates of the \ensuremath{i}th and \ensuremath{j}th
grid points for $\mathbf{x}_{i},\mathbf{x}_{j}\in\mathcal{X}$,
\begin{multline}
\mathbb{P}(\mathbf{p}_{\text{\ensuremath{t}}}={\bf x}_{i}|\mathbf{p}_{t-1}={\bf x}_{j})\\
\propto\begin{cases}
\exp(-\|{\bf x}_{i}-({\bf x}_{j}+\delta_{t}\bar{{\bf v}})\|_{2}^{2}) & \text{if }\|{\bf x}_{i}-{\bf x}_{j}\|_{2}^{2}\leq\eta\delta_{t}\\
0 & \text{otherwise}
\end{cases}\label{eq:mobility-transition-probability-1}
\end{multline}
where $\eta$ determines the maximum possible speed, $\bar{\mathbf{v}}$
is the average speed, and $\delta_{t}$ is the time slot duration.

A radio-map-embedded first-order \ac{ar} model \cite{VinJunHam:C24,SadRapKenAbh:C26}
for the dynamic of the channel $\mathbf{h}_{t}$ is considered as
\begin{equation}
{\bf h}_{t}=\gamma{\bf h}_{t-1}+\sqrt{1-\gamma^{2}}{\bf u}_{t}\label{eq:ar_channel_model}
\end{equation}
where ${\bf u}_{t}\sim\mathcal{CN}(0,{\bf C}({\bf p}_{t}))$, and
$\gamma\in[0,1]$ determines the level of temporal correlation between
the current channel ${\bf h}_{t}$ and the previous channel ${\bf h}_{t-1}$.
Additionally, the scaling factor $\sqrt{1-\gamma^{2}}$ ensures that
the disturbance term ${\bf u}_{t}$ maintains the appropriate variance
to preserve the statistical properties of the channel model. It is
worth highlighting that the statistics of the disturbance ${\bf u}_{t}$
depend on the user position $\mathbf{p}_{t}$ due to the geographic-specific
spatial correlation of the channel, and the statistics ${\bf C}({\bf p}_{t})$
are captured by the radio map. It can be verified that the stationary
distribution of ${\bf h}_{t}$ follows $\mathcal{CN}(0,{\bf C}({\bf p}_{t}))$
which is consistent with our radio map model.

At each time slot, the channel is probed using a small set of sensing
vectors, and the observation is given by
\begin{equation}
{\bf y}_{t}={\bf A}_{t}{\bf h}_{t}+{\bf n}_{t}\label{eq:observation_model}
\end{equation}
where ${\bf A}_{t}\in\mathbb{C}^{M\times N_{\mathrm{t}}}$ is a semi-unitary
sparse combining matrix satisfying ${\bf A}_{t}{\bf A}_{t}^{\text{H}}={\bf I}$
with $M\ll N_{\mathrm{t}}$, and ${\bf n}_{t}$ represents the measurement
noise. The noise is assumed to follow a zero-mean complex Gaussian
distribution, {\em i.e.}, ${\bf n}_{t}\sim\mathcal{CN}({\bf 0},\sigma_{\text{n}}^{2}{\bf I})$,
where $\sigma_{\text{n}}^{2}$ is the noise variance.

The observation model~(\ref{eq:observation_model}) provides a unified
representation for both uplink and downlink implementations. In the
uplink case, the \ac{bs} receives the signal ${\bf y}_{t}={\bf A}_{t}{\bf h}_{t}s+{\bf n}_{t}\in\mathbb{C}^{M}$
through the beamforming matrix ${\bf A}_{t}$, where the rows of ${\bf A}_{t}$
serve as combining vectors and $s$ denotes the reference symbol.
The parameter $M$ may correspond to the number of available \ac{rf}
chains at the \ac{bs}. We can simply take $s=1$ \ac{wlog} to yield
the observation model (\ref{eq:observation_model}). Note that $\mathbf{n}_{t}$
is the equivalent noise vector after combining the noise from the
$N_{\mathrm{t}}$ antennas using the semi-unitary matrix ${\bf A}_{t}\in\mathbb{C}^{M\times N_{\mathrm{t}}}$.
If $\tilde{\mathbf{n}}_{t}$ denotes the $N_{\mathrm{t}}\times1$
noise vector at the antennas, then $\mathbf{n}_{t}=\mathbf{A}_{t}\tilde{\mathbf{n}}_{t}$.
Since each row of $\mathbf{A}_{t}$ is orthonormal, both the elements
of $\mathbf{n}_{t}$ and those of $\tilde{\mathbf{n}}_{t}$ are \ac{iid}
with zero mean and the same variance. In the downlink case, consider
using $M$ orthogonal time-frequency resources for channel probing
at time slot $t$. For the $m$th time-frequency resource, a reference
symbol $s$ is beamformed using a vector $\mathbf{a}_{t,m}$, where
$\mathbf{a}_{t,m}^{\text{T}}$ corresponds to the $m$th row of $\mathbf{A}_{t}$.
As a result, the user obtains a noisy observation $y_{t,m}=\mathbf{h}_{t}^{\text{T}}\mathbf{a}_{t,m}s+n_{t,m}$,
and by stacking the $M$ observations $y_{t,m}$ into a vector $\mathbf{y}_{t}$,
the received signal model is identical to (\ref{eq:observation_model}).
In this case, the user may need to feed back the received pilot observations
$\mathbf{y}_{t}$ to the \ac{bs} or it is the user that performs
the channel tracking.

Note that both downlink and uplink transmissions involve feedback
signaling, such as scheduling grants and CSI feedback as specified
in the 5G standard. However, measuring the feedback overhead, which
is strongly implementation dependent, falls outside the scope of this
paper. Therefore, this paper intentionally considers only the $M$
pilot observations as the sole contributors to the received vector
$\mathbf{y}_{t}$ for CSI tracking.

\subsection{The CSI Tracking Problem}

\label{subsec:The-CSI-Tracking-Problem}

Denote $\mathcal{Y}_{t}=(\mathbf{y}_{1},\mathbf{y}_{2},\dots,\mathbf{y}_{\text{\ensuremath{t}}})$
as the sequence of sparse observations along the trajectory $\mathcal{P}_{t}=(\mathbf{p}_{1},\mathbf{p}_{2},\dots,\mathbf{p}_{t})$
based on sensing matrices $\mathcal{A}_{t}=(\mathbf{A}_{1},\mathbf{A}_{2},\dots,\mathbf{A}_{t})$,
where the \ac{csi} sequence is denoted as $\mathcal{H}_{t}=({\bf h}_{1},{\bf h}_{2},\dots,{\bf h}_{t})$.
The goal of this paper is to track the \ac{csi} $\mathcal{H}_{t}$
based on the sequence of sparse observations $\mathcal{Y}_{t}$. In
particular, we propose to recover the trajectory $\mathcal{P}_{t}$
as an intermediate step to assist the \ac{csi} tracking via the radio
map $\mathcal{M}=\{(\mathbf{x},{\bf C}({\bf x}):\mathbf{x}\in\mathcal{X}\}$.
In turn, we also need to construct or update the radio map $\mathcal{M}$
based on $\mathcal{Y}_{t}$ and $\mathcal{P}_{t}$.

\subsubsection{Radio-Map-Embedded CSI Tracking}

\label{subsec:Radio-Map-Embedded-CSI}

Most existing radio-map-embedded \ac{csi} estimation or beam-tracking
approaches require location information \cite{WanZhuLinChe:J24,WuZenJinZha:J23},
or an estimate of the location based on a one-shot \ac{csi} measurement
$\mathbf{y}_{t}$ \cite{WuZen:C23}. When a sequence of measurement
$\mathcal{Y}_{t}$ is available, one may estimate the entire \ac{csi}
sequence $\mathcal{H}_{t}$ by maximizing the joint probability $p(\mathcal{Y}_{t},\mathcal{H}_{t})$
or by minimizing the \ac{mse} for $\mathcal{H}_{t}$ based on the
joint distribution $p(\mathcal{Y}_{t},\mathcal{H}_{t})$. A classical
algorithm is \ac{kf} \cite{BroWanDas:J15,VinJunHam:C24}, but it
requires the channel to be temporally correlated with smooth statistics.
Specifically, a conventional \ac{kf} requires the covariance matrix
$\mathbf{C}(\mathbf{p}_{t})$ in (\ref{eq:ar_channel_model}) to be
continuous in space, without jumps. However, in a dense urban scenario,
the temporal and spatial correlation of the channel $\mathbf{h}_{t}$
can be destroyed by sudden signal blockage due to mobility, and therefore,
the process of the covariance matrix $\mathbf{C}(\mathbf{p}_{t})$
can be {\em discontinuous} and with jumps in space.

Therefore, we propose to develop a new \ac{csi} tracking strategy
based on the spatial correlation of the hidden trajectory $\mathcal{P}_{t}$,
which has to be {\em continuous} due to Newton's law. However, $\mathcal{P}_{t}$
is not observed but treated as a latent variable to be estimated.
Once estimated, one can track the \ac{csi} using the prior information
from the radio map $\mathcal{M}$, without assuming any correlation
between $\mathbf{C}(\mathbf{p}_{t})$ and $\mathbf{C}(\mathbf{p}_{t-1})$.
Mathematically, when a maximum log-likelihood criterion is considered,
the \ac{csi} tracking problem can be formulated as follows 
\begin{align*}
\mathscr{P}_{1}:\qquad\mathop{\mbox{maximize}}\limits_{\mathcal{H}_{T},\mathcal{P}_{T},\mathcal{A}_{T}} & \quad\log p(\mathcal{Y}_{T},\mathcal{H}_{T},\mathcal{P}_{T};\mathcal{A}_{T},\mathcal{M})\\
\text{subject to} & \quad\mathbf{p}_{t}\in\mathcal{X},\qquad t=1,2,\dots,T
\end{align*}
where $p(\mathcal{Y}_{T},\mathcal{H}_{T},\mathcal{P}_{T};\mathcal{A}_{T},\mathcal{M})$
denotes the joint probability of $\mathcal{Y}_{T}=(\mathbf{y}_{1},\mathbf{y}_{2},\dots,\mathbf{y}_{\text{\ensuremath{T}}})$,
$\mathcal{H}_{T}=({\bf h}_{1},{\bf h}_{2},\dots,{\bf h}_{T})$, and
$\mathcal{P}_{T}=(\mathbf{p}_{1},\mathbf{p}_{2},\dots,\mathbf{p}_{T})$,
given the sequence of sensing matrices $\mathcal{A}_{T}=(\mathbf{A}_{1},\mathbf{A}_{2},\dots,\mathbf{A}_{T})$
and the covariance matrices $\mathbf{C}(\mathbf{p}_{t})$ from the
radio map $\mathcal{M}$.

\subsubsection{Radio Map Construction}

\label{subsec:Joint-Radio-Map}

To construct the radio map $\mathcal{M}$, the anchor information
is needed, because there is no additional geographic information in
$\mathscr{P}_{1}$ except $\mathcal{M}$. Denote $\bm{\Theta}$ as
a collection of the \ac{bs} locations. Denote $f(\mathcal{Y}_{T},\mathcal{P}_{T};\bm{\Theta})$
as the value function to measure the likelihood of how the sequence
of observations $\mathcal{Y}_{T}$ matches with the trajectory $\mathcal{P}_{T}$
given the network topology $\bm{\Theta}$, such as the \ac{bs} locations.
Given the entire history of measurements, a topology-regularized maximum
log-likelihood problem for radio map construction is formulated as
follows
\begin{align*}
\mathscr{P}_{2}:\qquad\mathop{\mbox{maximize}}\limits_{\mathcal{H}_{T},\mathcal{P}_{T},\mathcal{M}} & \quad\log p(\mathcal{Y}_{T},\mathcal{H}_{T},\mathcal{P}_{T};\mathcal{M})\\
 & \quad\quad\quad-\mu f(\mathcal{Y}_{T},\mathcal{P}_{T};\bm{\Theta})\\
\text{subject to} & \quad\mathbf{p}_{t}\in\mathcal{X},\qquad t=1,2,\dots,T
\end{align*}
where $\mu>0$ serves as a regularization parameter. The radio map
variable $\mathcal{M}$ in $\mathscr{P}_{2}$ is constructed via estimating
the covariance matrices $\mathbf{C}(\mathbf{x})$ for all $\mathbf{x}\in\mathcal{X}$.

\section{Radio-Map-Embedded CSI Tracking with A Switching Kalman Filter Framework}

\label{sec:Radio-map-assisted-CSI-Tracking}

\begin{figure}
\begin{centering}
\includegraphics[width=1\columnwidth]{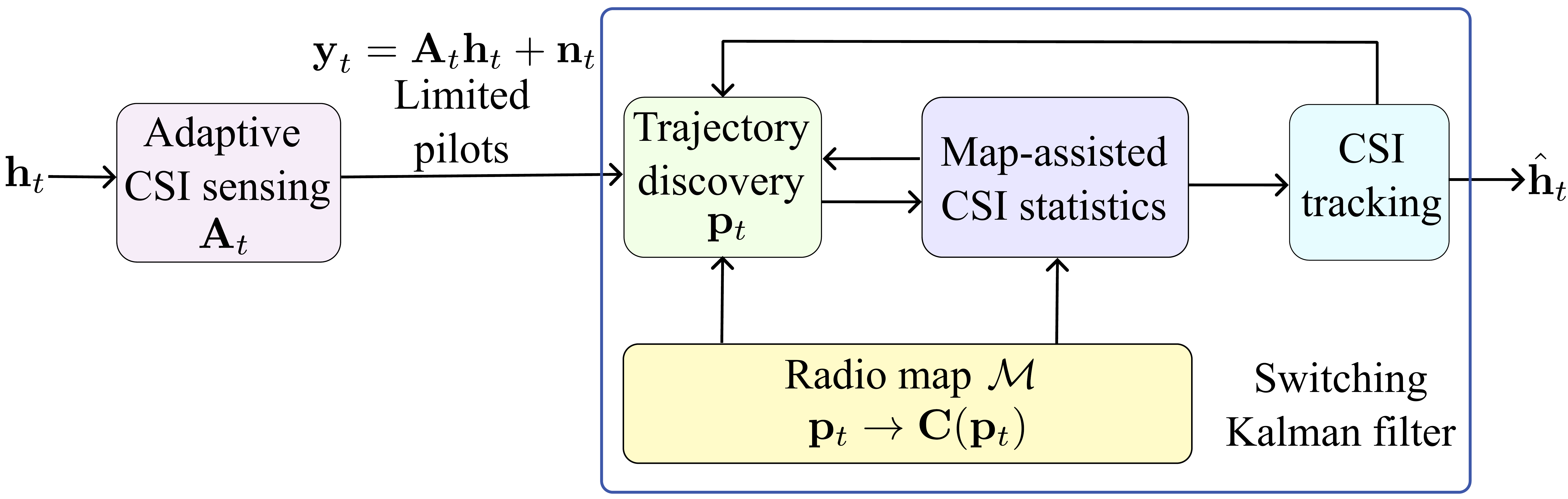}
\par\end{centering}
\caption{\label{fig:Diagram-of-the-radio-map-assisted-CSI-tracking}Diagram
of the radio-map-embedded CSI tracking framework.}
\end{figure}

In this section, we are interested in a tracking approach to solve
problem $\mathscr{P}_{1}$. Specifically, we develop a radio-map-embedded
\ac{skf} framework for \ac{csi} tracking as shown in Fig.~\ref{fig:Diagram-of-the-radio-map-assisted-CSI-tracking}.
First, given a sequence of sensing matrices $\mathcal{A}_{t}$, the
position $\mathbf{p}_{t}$ and the \ac{csi} $\mathbf{h}_{t}$ are
jointly estimated and tracked, where an \ac{skf} is applied to track
$\mathbf{h}_{t}$ with switching dynamics based on the distribution
of the position $\mathbf{p}_{t}$ and the radio map $\mathcal{M}$,
and in turn, the position $\mathbf{p}_{t}$ is estimated and tracked
based on the estimate of $\mathbf{h}_{t}$. Second, the sensing matrix
$\mathbf{A}_{t}$ is dynamically adjusted by minimizing the differential
entropy of the channel.

As the radio map $\mathcal{M}$ is given in this section, to simplify
the notation, we omit the symbols $\mathcal{A}_{T}$ and $\mathcal{M}$
in the joint probability $p(\mathcal{Y}_{T},\mathcal{H}_{T},\mathcal{P}_{T};\mathcal{A}_{T},\mathcal{M})$,
wherever it does not affect the clarity from the context.

\subsection{Problem Reformulation}

\label{subsec:Problem-Reformulation}

Using Bayes' theorem, the joint probability $p(\mathcal{Y}_{T},\mathcal{H}_{T},\mathcal{P}_{T})$
is explicitly expressed in the following proposition.
\begin{prop}[Factorization of $p(\mathcal{Y}_{T},\mathcal{H}_{T},\mathcal{P}_{T})$]
\label{prop:Factorization-of-logp}The joint probability $p(\mathcal{Y}_{T},\mathcal{H}_{T},\mathcal{P}_{T})$
can be factorized as
\begin{align}
p(\mathcal{Y}_{T},\mathcal{H}_{T},\mathcal{P}_{T})= & \prod_{t=1}^{T}p(\mathbf{y}_{\text{\ensuremath{t}}}|\mathbf{h}_{t})\prod_{t=2}^{T}p({\bf h}_{t}|{\bf h}_{t-1},{\bf p}_{t})\nonumber \\
 & \times p({\bf h}_{1})\times\prod_{t=2}^{T}\mathbb{P}(\mathbf{p}_{\text{\ensuremath{t}}}|\mathbf{p}_{t-1})\mathbb{P}(\mathbf{p}_{1}).\label{eq:factorization-F-YHP}
\end{align}
\end{prop}
\begin{proof}
See Appendix~\ref{sec:Proof-of-Proposition-Factorization}.
\end{proof}

With the factorization~(\ref{eq:factorization-F-YHP}) and given
the sensing matrices $\mathcal{A}_{T}$, problem $\mathscr{P}_{1}$
becomes
\begin{align*}
\mathscr{P}_{1}':\qquad\mathop{\mbox{maximize}}\limits_{\mathcal{H}_{T},\mathcal{P}_{T}} & \quad\sum_{t=1}^{T}\log p(\mathbf{y}_{\text{\ensuremath{t}}}|\mathbf{h}_{t})+\log p({\bf h}_{1})\\
 & \quad+\sum_{t=2}^{T}\log p({\bf h}_{t}|{\bf h}_{t-1},{\bf p}_{t})\\
 & \quad+\log\mathbb{P}(\mathbf{p}_{1})+\sum_{t=2}^{T}\log\mathbb{P}(\mathbf{p}_{\text{\ensuremath{t}}}|\mathbf{p}_{t-1})\\
\text{subject to} & \quad\mathbf{p}_{t}\in\mathcal{X},\qquad t=1,2,\dots,T.
\end{align*}

\subsection{Radio-Map-Embedded Switching Kalman Filtering}

\label{subsec:Switch-Kalman-Filtering}

Here, we develop an \ac{skf} approach to jointly track the positions
and the \ac{csi}. We first find the distribution of the position
$\mathbf{p}_{t}$ given the measurements up to time $t$ and the trajectory
up to time $t-1$. Then, with the radio map~$\mathcal{M}$ and the
conditional distribution of $\mathbf{p}_{t}$, the \ac{csi} $\mathbf{h}_{t}$
is recursively tracked leveraging an \ac{skf}. Finally, the position
$\mathbf{p}_{t}$ is estimated based on the maximum likelihood criterion.

\subsubsection{Distribution of the Position}

\label{subsec:Mode-Switching-1}

Let $\pi_{t}(\mathbf{p}_{t})\triangleq p(\mathbf{p}_{t}|\mathcal{Y}_{t},\mathcal{P}_{t-1})$
be the distribution of the position at time $t$ given the measurements
up to time $t$ and the trajectory up to time $t-1$. Using Bayes'
theorem, for $t=2,3,\dots T$, $\pi_{t}(\mathbf{p}_{t})$ is derived
as
\begin{align}
\pi_{t}(\mathbf{p}_{t}) & =p(\mathbf{p}_{t}|\mathcal{Y}_{t},\mathcal{P}_{t-1})\nonumber \\
 & =\frac{p(\mathcal{Y}_{t},\mathcal{P}_{t-1},\mathbf{p}_{t})}{p(\mathcal{Y}_{t},\mathcal{P}_{t-1})}\label{eq:7-a}\\
 & =\frac{p(\mathcal{Y}_{t},\mathcal{P}_{t})}{p(\mathcal{Y}_{t},\mathcal{P}_{t-1})}\label{eq:7-b}\\
 & =\frac{p(\mathbf{y}_{t}|\mathcal{Y}_{t-1},\mathcal{P}_{t})p(\mathbf{p}_{t}|\mathcal{Y}_{t-1},\mathcal{P}_{t-1})p(\mathcal{Y}_{t-1},\mathcal{P}_{t-1})}{p(\mathbf{y}_{t}|\mathcal{Y}_{t-1},\mathcal{P}_{t-1})p(\mathcal{Y}_{t-1},\mathcal{P}_{t-1})}\label{eq:7-c}\\
 & =\frac{p(\mathbf{y}_{t}|\mathbf{p}_{t})\mathbb{P}(\mathbf{p}_{t}|\mathbf{p}_{t-1})}{\sum_{\mathbf{x}\in\mathcal{X}}p({\bf y}_{t}|\mathbf{x})\mathbb{P}(\mathbf{x}|\mathbf{p}_{t-1})}\label{eq:pi-pt}
\end{align}
where equation~(\ref{eq:7-a}) follows from the definition of conditional
probability, equation~(\ref{eq:7-b}) uses the identity that $\mathcal{P}_{t}=(\mathbf{p}_{1},\mathbf{p}_{2},\dots,\mathbf{p}_{t})$
and $\mathcal{P}_{t-1}=(\mathbf{p}_{1},\mathbf{p}_{2},\dots,\mathbf{p}_{t-1})$,
equation~(\ref{eq:7-c}) holds since both $p(\mathcal{Y}_{t},\mathcal{P}_{t})$
and $p(\mathcal{Y}_{t},\mathcal{P}_{t-1})$ can be derived in a similar
manner to (\ref{eq:joint_prop_y_h}), and equation~(\ref{eq:pi-pt})
holds because $\mathbf{y}_{t}$ depends on $\mathbf{p}_{t}$ which
depends on $\mathbf{p}_{t-1}$. The conditional probability $p({\bf y}_{t}|{\bf p}_{t})$
can be derived based on the observation model~(\ref{eq:observation_model})
and the channel distribution $\mathbf{h}_{t}|\mathbf{p}_{t}\sim\mathcal{CN}(0,{\bf C}({\bf p}_{t}))$
indicated in~(\ref{eq:ar_channel_model}) as
\begin{equation}
p({\bf y}_{t}|{\bf p}_{t})=\frac{\exp(-\frac{1}{2}\mathbf{y}_{t}^{\mathrm{H}}({\bf A}_{t}{\bf C}({\bf p}_{t}){\bf A}_{t}^{\text{H}}+\sigma_{\text{n}}^{2}{\bf I})^{-1}\mathbf{y}_{t})}{\sqrt{(2\pi)^{M}|{\bf A}_{t}{\bf C}({\bf p}_{t}){\bf A}_{t}^{\text{H}}+\sigma_{\text{n}}^{2}{\bf I}|}}.
\end{equation}

\subsubsection{Radio-Map-Embedded \ac{csi} Tracking}

\label{subsec:-Driven-Unbiased-Channel-Estimate}

First, we derive a \ac{csi} prediction and its corresponding error
covariance at time $t$ based on the observations $\mathcal{Y}_{t-1}$
up to time $t-1$.

Define $\hat{\mathbf{h}}_{t-1}\triangleq\mathbb{E}\{\mathbf{h}_{t-1}|\mathcal{Y}_{t-1}\}$
as the Bayesian estimator of the \ac{csi} $\mathbf{h}_{t-1}$ based
on the observations $\mathcal{Y}_{t-1}$ up to time $t-1$. Define
$\hat{\mathbf{h}}_{t|t-1}\triangleq\mathbb{E}\{\mathbf{h}_{t}|\hat{\mathbf{h}}_{t-1}\}$
as a prediction of the \ac{csi} $\mathbf{h}_{t}$ based on $\hat{\mathbf{h}}_{t-1}$.
Using the \ac{csi} dynamic model~(\ref{eq:ar_channel_model}), we
have 
\begin{equation}
\hat{\mathbf{h}}_{t|t-1}=\mathbb{E}\{\mathbf{h}_{t}|\hat{\mathbf{h}}_{t-1}\}=\gamma\hat{{\bf h}}_{t-1}.\label{eq:channel-prediction}
\end{equation}
Define $\mathbf{e}_{t|t-1}\triangleq\mathbf{h}_{t}-\hat{\mathbf{h}}_{t|t-1}$
as the channel prediction error of $\mathbf{h}_{t}$ based on the
observations $\mathcal{Y}_{t-1}$ up to time $t-1$. Substituting
(\ref{eq:channel-prediction}) into $\mathbf{e}_{t|t-1}$, we can
obtain
\begin{alignat}{1}
\mathbf{e}_{t|t-1} & \triangleq\mathbf{h}_{t}-\hat{\mathbf{h}}_{t|t-1}\nonumber \\
 & =\gamma\mathbf{h}_{t-1}+\sqrt{1-\gamma^{2}}{\bf u}_{t}-\gamma\hat{{\bf h}}_{t-1}\nonumber \\
 & =\gamma\mathbf{e}_{t-1}+\sqrt{1-\gamma^{2}}{\bf u}_{t}\label{eq:prior-channel-estimate-error}
\end{alignat}
where $\mathbf{e}_{t-1}$ is the channel estimation error generated
from the Bayesian estimator $\hat{\mathbf{h}}_{t-1}$ , {\em i.e.},
$\mathbf{e}_{t-1}\triangleq\mathbf{h}_{t-1}-\hat{{\bf h}}_{t-1}$.

Following the channel prediction error $\mathbf{e}_{t|t-1}$ in (\ref{eq:prior-channel-estimate-error}),
the error covariance of $\mathbf{e}_{t|t-1}$ is derived as
\begin{alignat}{1}
\mathbf{Q}_{t|t-1} & \triangleq\mathbb{E}\{\mathbf{e}_{t|t-1}\mathbf{e}_{t|t-1}^{\text{H}}\}\nonumber \\
 & =\gamma^{2}\mathbb{E}\{\mathbf{e}_{t-1}\mathbf{e}_{t-1}^{\text{H}}\}+(1-\gamma^{2})\mathbb{E}\{\mathbf{u}_{t}\mathbf{u}_{t}^{\text{H}}|\mathbf{p}_{t-1}\}\nonumber \\
 & =\gamma^{2}\mathbf{Q}_{t-1}+(1-\gamma^{2})\mathbb{E}\{\mathbf{u}_{t}\mathbf{u}_{t}^{\text{H}}|\mathbf{p}_{t-1}\}\label{eq:prior-error-covariance}
\end{alignat}
where $\mathbf{Q}_{t-1}\triangleq\mathbb{E}\{\mathbf{e}_{t-1}\mathbf{e}_{t-1}^{\text{H}}\}$
is the error covariance of $\mathbf{e}_{t-1}$, and the term $\mathbb{E}\{\mathbf{u}_{t}\mathbf{u}_{t}^{\mathrm{H}}|\mathbf{p}_{t-1}\}$
can be computed based on the distribution of the position $\mathbf{p}_{t}$
given $\mathbf{p}_{t-1}$ in (\ref{eq:pi-pt}) and the radio map~$\mathcal{M}$
as follows
\begin{align}
\mathbb{E}\{\mathbf{u}_{t}\mathbf{u}_{t}^{\mathrm{H}}|\mathbf{p}_{t-1}\} & =\mathbb{E}_{\mathbf{p}_{t}\in\mathcal{X}}\{\mathbb{E}\{\mathbf{u}_{t}\mathbf{u}_{t}^{\mathrm{H}}|\mathbf{p}_{t}\}|\mathbf{p}_{t-1}\}\nonumber \\
 & =\sum_{\mathbf{p}_{t}\in\mathcal{X}}\pi_{t}({\bf p}_{t})\times\mathbf{C}(\mathbf{p}_{t})\label{eq:expectation-ut}
\end{align}
where recall from~(\ref{eq:ar_channel_model}) that $\mathbf{u}_{t}|\mathbf{p}_{t}\sim\mathcal{CN}(0,\mathbf{C}(\mathbf{p}_{t}))$
implying $\mathbb{E}\{\mathbf{u}_{t}\mathbf{u}_{t}^{\mathrm{H}}|\mathbf{p}_{t}\}=\mathbf{C}(\mathbf{p}_{t})$,
and $\pi_{t}({\bf p}_{t})$ captures the location distribution at
time $t$ based on observations $\mathcal{Y}_{t}$ up to time $t$
and the user trajectory $\mathcal{P}_{t-1}$ up to time $t-1$. It
is observed that the radio map $\mathcal{M}$ plays a critical role
in constructing the distribution of the hidden position variable $\mathbf{p}_{t}$
and the error covariance $\mathbf{Q}_{t|t-1}$ for the channel prediction.

Second, an update equation for the Bayesian estimator $\hat{\mathbf{h}}_{t}\triangleq\mathbb{E}\{\mathbf{h}_{t}|\mathcal{Y}_{t}\}$
based on the observations~$\mathcal{Y}_{t}$ up to time $t$ can
be formulated as \cite{WelBis:B95}
\begin{align}
\hat{\mathbf{h}}_{t} & =\hat{\mathbf{h}}_{t|t-1}+\mathbf{K}_{t}(\mathbf{y}_{t}-\mathbf{A}_{t}\hat{\mathbf{h}}_{t|t-1})\label{eq:channel-estimate-equation}
\end{align}
where $\mathbf{K}_{t}$ is a coefficient of the measurement residual
$\mathbf{y}_{t}-\mathbf{A}_{t}\hat{\mathbf{h}}_{t|t-1}$. The optimal
coefficient $\mathbf{K}_{t}$ that minimizes the mean squared error
of the estimator $\hat{\mathbf{h}}_{t}$ is given by \cite{WelBis:B95}
\begin{equation}
{\bf K}_{t}=\mathbf{Q}_{t|t-1}{\bf A}_{t}^{\text{H}}({\bf A}_{t}\mathbf{Q}_{t|t-1}{\bf A}_{t}^{\text{H}}+\sigma_{\mathrm{n}}^{2}{\bf I})^{-1}.\label{eq:Kalman-gain}
\end{equation}

It is observed from (\ref{eq:Kalman-gain}) that the optimal $\mathbf{K}_{t}$
relies on the error covariance $\mathbf{Q}_{t|t-1}$, which is partially
determined by the channel spatial statistics $\mathbb{E}\{\mathbf{u}_{t}\mathbf{u}_{t}^{\mathrm{H}}|\mathbf{p}_{t-1}\}$,
and thus the radio map~$\mathcal{M}$.

Based on (\ref{eq:channel-estimate-equation}) and (\ref{eq:Kalman-gain}),
the covariance of the estimation error $\mathbf{e}_{t}\triangleq\mathbf{h}_{t}-\hat{\mathbf{h}}_{t}$
can be derived as
\begin{alignat}{1}
\mathbf{Q}_{t} & \triangleq\mathbb{E}\{(\mathbf{h}_{t}-\hat{\mathbf{h}}_{t})(\mathbf{h}_{t}-\hat{\mathbf{h}}_{t})^{\text{H}}\}\nonumber \\
 & =({\bf I}-{\bf K}_{t}{\bf A}_{t})\mathbf{Q}_{t|t-1}({\bf I}-{\bf K}_{t}{\bf A}_{t})^{\text{H}}+\mathbf{K}_{t}\mathbf{K}_{t}^{\text{H}}/\sigma_{\text{n}}^{2}\nonumber \\
 & =({\bf I}-{\bf K}_{t}{\bf A}_{t})\mathbf{Q}_{t|t-1}\label{eq:update-error-covariance}
\end{alignat}
which completes the recursive filter.

\subsubsection{User Position Tracking}

\label{subsec:User-Position-Tracking}

Based on the observation sequence $\mathcal{Y}_{t}$, the \ac{csi}
sequence $\mathcal{H}_{t}$ up to time $t$, and the user trajectory
$\mathcal{P}_{t-1}$, the user position at time $t=2,3,\dots,T$ can
be tracked by maximizing the conditional probability $\bar{\pi}_{t}(\mathbf{p}_{t})\triangleq p(\mathbf{p}_{t}|\mathcal{Y}_{t},\mathcal{H}_{t},\mathcal{P}_{t-1})$.
Similar to (\ref{eq:pi-pt}), $\bar{\pi}_{t}(\mathbf{p}_{t})$ can
be derived as
\begin{align}
\bar{\pi}_{t}(\mathbf{p}_{t}) & \triangleq p(\mathbf{p}_{t}|\mathcal{Y}_{t},\mathcal{H}_{t},\mathcal{P}_{t-1})\nonumber \\
 & =\frac{p(\mathcal{Y}_{t},\mathcal{H}_{t},\mathcal{P}_{t-1},\mathbf{p}_{t})}{p(\mathcal{Y}_{t},\mathcal{H}_{t},\mathcal{P}_{t-1})}\label{eq:19-a}\\
 & =\frac{p(\mathcal{Y}_{t},\mathcal{H}_{t},\mathcal{P}_{t})}{p(\mathcal{Y}_{t},\mathcal{H}_{t},\mathcal{P}_{t-1})}\label{eq:19-b}\\
 & =\frac{p(\mathbf{h}_{t}|\mathbf{h}_{t-1},\mathbf{p}_{t})\mathbb{P}(\mathbf{p}_{t}|\mathbf{p}_{t-1})}{\sum_{\mathbf{x}\in\mathcal{X}}p(\mathbf{h}_{t}|\mathbf{h}_{t-1},\mathbf{x})\mathbb{P}(\mathbf{x}|\mathbf{p}_{t-1})}\label{eq:bar-pi-pt}
\end{align}
where equation~(\ref{eq:19-a}) follows from the definition of conditional
probability, equation~(\ref{eq:19-b}) uses the identity that $\mathcal{P}_{t}=(\mathbf{p}_{1},\mathbf{p}_{2},\dots,\mathbf{p}_{t})$
and $\mathcal{P}_{t-1}=(\mathbf{p}_{1},\mathbf{p}_{2},\dots,\mathbf{p}_{t-1})$,
and equation~(\ref{eq:bar-pi-pt}) holds because the joint distribution
$p(\mathcal{Y}_{t},\mathcal{H}_{t},\mathcal{P}_{t})$ is given in
(\ref{eq:joint-probability-Y-P-H}), and $p(\mathcal{Y}_{t},\mathcal{H}_{t},\mathcal{P}_{t-1})$
is obtained by marginalizing $p(\mathcal{Y}_{t},\mathcal{H}_{t},\mathcal{P}_{t})$
over ${\bf p}_{t}$. In addition, the conditional distribution $\mathbf{h}_{t}|\mathbf{h}_{t-1},\mathbf{p}_{t}\sim\mathcal{CN}(\gamma{\bf h}_{t-1},(1-\gamma^{2})\mathbf{C}(\mathbf{p}_{t}))$
follows from (\ref{eq:ar_channel_model}).

Based on (\ref{eq:bar-pi-pt}), the user position at time $t=2,3,\dots,T$
can be obtained as 
\begin{align}
\hat{\mathbf{p}}_{t} & =\mathop{\mbox{argmax}}\limits_{\mathbf{p}_{t}\in\mathcal{X}}\quad p(\mathbf{h}_{t}|\mathbf{h}_{t-1},\mathbf{p}_{t})\mathbb{P}(\mathbf{p}_{t}|\mathbf{p}_{t-1})\nonumber \\
 & =\mathop{\mbox{argmax}}\limits_{\mathbf{p}_{t}\in\mathcal{X}}\frac{\exp(-\frac{1}{2}\boldsymbol{\Delta}_{t}^{\mathrm{H}}((1-\gamma^{2})\mathbf{C}(\mathbf{p}_{t}))^{-1}\boldsymbol{\Delta}_{t})}{\sqrt{(2\pi(1-\gamma^{2}))^{N_{\mathrm{t}}}|\mathbf{C}(\mathbf{p}_{t})|}}\mathbb{P}(\mathbf{p}_{t}|\mathbf{p}_{t-1})\label{eq:tracking-user-position}
\end{align}
where $\boldsymbol{\Delta}_{t}\triangleq\mathbf{h}_{t}-\gamma{\bf h}_{t-1}$.

The overall algorithm of the radio-map-embedded \ac{skf} is summarized
in Algorithm~\ref{alg:radio-map-assisted-switching-Kalman-filter-csi-tracking}.
The complexity of Algorithm~\ref{alg:radio-map-assisted-switching-Kalman-filter-csi-tracking}
is analyzed as follows. At each time step, the computational complexity
for \ac{csi} tracking is $O(|\mathcal{X}|N_{\text{t}}^{2}+N_{\text{t}}^{2}M+M^{3}+N_{t}^{3})$.
Specifically the term $O(|\mathcal{X}|N_{\text{t}}^{2})$ arises from
computing the spatial statistics of \ac{csi} in (\ref{eq:expectation-ut}),
$O(N_{\text{t}}^{2}M+M^{3})$ results from matrix operations in \ac{csi}
estimation, such as computing $\mathbf{Q}_{t|t-1}{\bf A}_{t}^{\text{H}}$
in (\ref{eq:Kalman-gain}), and $O(N_{t}^{3})$ is due to the matrix
inversion involved in user localization in (\ref{eq:tracking-user-position}).

\begin{algorithm}
\textbf{Input:} Radio maps $\mathcal{M}$, noise variance $\sigma_{\text{n}}^{2}$,
and an initialized position $\hat{\mathbf{p}}_{1}$ (such as from
a uniform prior).
\begin{enumerate}
\item At time $t=1$, randomly generate $\mathbf{A}_{1}$ to observe $\mathbf{y}_{1}$.
Initialize $\hat{{\bf h}}_{1}\sim\mathcal{CN}(0,\mathbf{C}(\hat{\mathbf{p}}_{1}))$
and ${\bf Q}_{1}=(1-\gamma^{2})\sum_{\mathbf{p}_{1}\in\mathcal{X}}\pi_{1}({\bf p}_{1})\times\mathbf{C}(\mathbf{p}_{1})$,
where $\pi_{1}({\bf p}_{1})=p(\mathbf{p}_{1}|\mathbf{y}_{1})$.
\item For each $t>1$:
\begin{enumerate}
\item \label{enu:Prior-channel-estimate}Channel prediction based on observations
$\mathcal{Y}_{t-1}$:
\begin{enumerate}
\item $\hat{{\bf h}}_{t|t-1}=\gamma\hat{{\bf h}}_{t-1}$.
\item \label{enu:approximate-of-Qtt-1}${\bf Q}_{t|t-1}=\gamma^{2}\mathbf{Q}_{t-1}+(1-\gamma^{2})\sum_{\mathbf{p}_{t}\in\mathcal{X}}\mathbb{P}(\mathbf{p}_{t}|\mathbf{p}_{t-1}=\hat{\mathbf{p}}_{t-1})\times\mathbf{C}(\mathbf{p}_{t}),$
where $\mathbb{P}(\mathbf{p}_{t}|\mathbf{p}_{t-1}=\hat{\mathbf{p}}_{t-1})$
is given in (\ref{eq:mobility-transition-probability-1}) and $\mathbf{C}(\mathbf{p}_{t})$
is embedded in the radio map~$\mathcal{M}$.
\end{enumerate}
\item \label{enu:Channel-measurement}Design $\mathbf{A}_{t}$ using (\ref{eq:optimal-solution-A})
based on ${\bf Q}_{t|t-1}$ and perform channel sensing using $\mathbf{A}_{t}$
to obtain ${\bf y}_{t}$.
\item \label{enu:Refinement-of-Q_t-t-1}Refine ${\bf Q}_{t|t-1}$ as ${\bf Q}_{t|t-1}=\gamma^{2}\mathbf{Q}_{t-1}+(1-\gamma^{2})\mathbb{E}\{\mathbf{u}_{t}\mathbf{u}_{t}^{\text{H}}|\mathbf{p}_{t-1}=\hat{\mathbf{p}}_{t-1}\},$
where $\mathbb{E}\{\mathbf{u}_{t}\mathbf{u}_{t}^{\text{H}}|\mathbf{p}_{t-1}=\hat{\mathbf{p}}_{t-1}\}=\sum_{\mathbf{p}_{t}\in\mathcal{X}}\pi_{t}({\bf p}_{t})\times\mathbf{C}(\mathbf{p}_{t})$,
in which $\pi_{t}({\bf p}_{t})$ is calculated using (\ref{eq:pi-pt}).
\item Channel estimation based on observations $\mathcal{Y}_{t}$:
\begin{enumerate}
\item ${\bf K}_{t}=\mathbf{Q}_{t|t-1}{\bf A}_{t}^{\text{H}}({\bf A}_{t}\mathbf{Q}_{t|t-1}{\bf A}_{t}^{\text{H}}+\sigma_{\mathrm{n}}^{2}{\bf I})^{-1}.$
\item $\hat{\mathbf{h}}_{t}=\hat{\mathbf{h}}_{t|t-1}+\mathbf{K}_{t}(\mathbf{y}_{t}-\mathbf{A}_{t}\hat{\mathbf{h}}_{t|t-1})$.
\item ${\bf Q}_{t}=({\bf I}-{\bf K}_{t}{\bf A}_{t})\mathbf{Q}_{t|t-1}$.
\end{enumerate}
\item Estimate the user position $\hat{\mathbf{p}}_{t}=\text{argmax}_{\mathbf{p}_{t}\in\mathcal{X}}p(\hat{\mathbf{h}}_{t}|\hat{\mathbf{h}}_{t-1},\mathbf{p}_{t})\mathbb{P}(\mathbf{p}_{t}|\hat{\mathbf{p}}_{t-1})$.
\item Output $\hat{\mathbf{h}}_{t}$ and $\hat{\mathbf{p}}_{t}$.
\end{enumerate}
\end{enumerate}
\caption{\foreignlanguage{american}{Radio-Map-Embedded \ac{skf} for Online \ac{csi} Tracking}}

\label{alg:radio-map-assisted-switching-Kalman-filter-csi-tracking}
\end{algorithm}

\subsection{Entropy-based Sensing Matrix Adaptation}

\label{subsec:Sensing-Matrix-Adaption}

As can be seen in (\ref{eq:channel-estimate-equation}), the update
of the channel estimator $\hat{{\bf h}}_{t}$ from $\hat{\mathbf{h}}_{t|t-1}$
relies on the observation $\mathbf{y}_{t}$ which is influenced by
the sensing matrix $\mathbf{A}_{t}$. When no prior information is
available, a random sensing matrix $\mathbf{A}_{t}$ may be used.
However, in the scenario of tracking, prior information is available
for optimizing the sensing matrix $\mathbf{A}_{t}$ for better \ac{csi}
tracking. To achieve this, we investigate the differential entropy
of $\hat{{\bf h}}_{t}$ \ac{wrt} $\mathbf{A}_{t}$.

\subsubsection{Differential Entropy}

\label{subsec:Differential-Entropy-Minimization-Problem}

First, to evaluate the uncertainty of the channel estimator $\hat{{\bf h}}_{t}$,
we derive the differential entropy of $\hat{\mathbf{h}}_{t}$ \ac{wrt}
the error covariance $\mathbf{Q}_{t}$.

From (\ref{eq:update-error-covariance}), the error covariance $\mathbb{E}\{(\mathbf{h}_{t}-\hat{\mathbf{h}}_{t})(\mathbf{h}_{t}-\hat{\mathbf{h}}_{t})^{\text{H}}\}$
of $\hat{\mathbf{h}}_{t}$ given the measurements $\mathcal{Y}_{t}$
up to time $t$ is captured in $\mathbf{Q}_{t}$, which implies that
$\hat{\mathbf{h}}_{t}\sim\mathcal{CN}(\mathbf{h}_{t},\mathbf{Q}_{t}).$
Following this, the differential entropy of $\hat{\mathbf{h}}_{t}$
is derived as \cite[Theorem 8.4.1]{ThoTho:B06}
\begin{align}
q(\hat{\mathbf{h}}_{t}) & =\frac{N_{\mathrm{t}}}{2}(1+\log(2\pi))+\frac{1}{2}\log|\mathbf{Q}_{t}|.\label{eq:differential-entropy-hat-ht}
\end{align}

Similarly, based on (\ref{eq:channel-prediction}) and (\ref{eq:prior-error-covariance}),
the differential entropy of $\hat{\mathbf{h}}_{t|t-1}$ is
\begin{equation}
q(\hat{\mathbf{h}}_{t|t-1})=\frac{N_{\mathrm{t}}}{2}(1+\log(2\pi))+\frac{1}{2}\log|\mathbf{Q}_{t|t-1}|.\label{eq:differential-entropy-hat-h-t-1}
\end{equation}

Based on (\ref{eq:differential-entropy-hat-ht}) and (\ref{eq:differential-entropy-hat-h-t-1}),
we can express the differential entropy $q(\hat{\mathbf{h}}_{t})$
given $q(\hat{\mathbf{h}}_{t|t-1})$ as
\begin{equation}
q(\hat{\mathbf{h}}_{t})=q(\hat{\mathbf{h}}_{t|t-1})+\frac{1}{2}\left(\log|\mathbf{Q}_{t}|-\log|\mathbf{Q}_{t|t-1}|\right).\label{eq:reformulate-q-hat-ht}
\end{equation}

Next, we minimize $q(\hat{\mathbf{h}}_{t})$ \ac{wrt} the sensing
matrix $\mathbf{A}_{t}$, which is embedded in $\mathbf{Q}_{t}$ as
indicated in (\ref{eq:update-error-covariance}). Given $q(\hat{\mathbf{h}}_{t|t-1})$
and (\ref{eq:reformulate-q-hat-ht}), minimizing $q(\hat{\mathbf{h}}_{t})$
is equivalent to minimizing $q(\hat{\mathbf{h}_{t}})-q(\hat{\mathbf{h}}_{t|t-1})=\frac{1}{2}\left(\log|\mathbf{Q}_{t}|-\log|\mathbf{Q}_{t|t-1}|\right)$,
which leads to the following optimization problem
\begin{align}
\mathop{\mbox{minimize}}\limits_{\mathbf{A}_{t}}\quad & \frac{1}{2}\left(\log|\mathbf{Q}_{t}|-\log|\mathbf{Q}_{t|t-1}|\right)\label{eq:problem-maximizing-entropy}\\
\text{subject to}\quad & \mathbf{A}_{t}\mathbf{A}_{t}^{\text{H}}=\mathbf{I}.\nonumber 
\end{align}

\subsubsection{Solution for the Adaptive Sensing Matrix $\mathbf{A}_{t}$}

\label{subsec:Solution-for-the-Adaptive-Sening-Matrix-A-t}

To solve Problem~(\ref{eq:problem-maximizing-entropy}), we first
explicitly derive the expression of $\frac{1}{2}(\log|\mathbf{Q}_{t}|-\log|\mathbf{Q}_{t|t-1}|)$
\ac{wrt} $\mathbf{A}_{t}$. Although we have derived $\mathbf{Q}_{t|t-1}$
and $\mathbf{Q}_{t}$ in (\ref{eq:prior-error-covariance}) and (\ref{eq:update-error-covariance}),
it is not straightforward to derive $\frac{1}{2}(\log|\mathbf{Q}_{t}|-\log|\mathbf{Q}_{t|t-1}|)$
from them because $\mathbf{Q}_{t}$, $\mathbf{Q}_{t|t-1}$, $\mathbf{Q}_{t-1}$
are recursively coupled. Thus, we provide the following proposition
to offer an alternative yet equivalent transition from $\mathbf{Q}_{t|t-1}$
to $\mathbf{Q}_{t}$ of (\ref{eq:update-error-covariance}).
\begin{prop}[Alternative error propagation process]
\label{prop:error-propagation-equation}The error covariance $\mathbf{Q}_{t}$
follows the following error propagation process
\begin{equation}
\mathbf{Q}_{t}^{-1}=\mathbf{Q}_{t|t-1}^{-1}+\mathbf{A}_{t}^{\mathrm{H}}(\sigma_{\mathrm{n}}^{2}\mathbf{I})^{-1}\mathbf{A}_{t}.\label{eq:alternative-error-covariance-update}
\end{equation}
\end{prop}
\begin{proof}
See Appendix~\ref{sec:Proof-of-Proposition-Error-Propagation-Equation}.
\end{proof}

Taking the determinant of (\ref{eq:alternative-error-covariance-update})
yields
\begin{align}
|\mathbf{Q}_{t}^{-1}| & =|\mathbf{Q}_{t|t-1}^{-1}+\mathbf{A}_{t}^{\text{H}}(\sigma_{\mathrm{n}}^{2}\mathbf{I})^{-1}\mathbf{A}_{t}|\nonumber \\
 & =|\mathbf{Q}_{t|t-1}^{-1}||\mathbf{I}+\mathbf{Q}_{t|t-1}^{-1}\mathbf{A}_{t}^{\text{H}}(\sigma_{\mathrm{n}}^{2}\mathbf{I})^{-1}\mathbf{A}_{t}|\label{eq:use-det-lemma}\\
 & =|\mathbf{Q}_{t|t-1}^{-1}||\mathbf{I}+\sigma_{\text{n}}^{-2}\mathbf{A}_{t}\mathbf{Q}_{t|t-1}^{-1}\mathbf{A}_{t}^{\text{H}}|\label{eq:use-det-lemma2}
\end{align}
where equation~(\ref{eq:use-det-lemma}) utilizes the multiplicative
property of determinant, {\em i.e.}, $|\mathbf{XY}|=|\mathbf{X}||\mathbf{Y}|$,
and equation~(\ref{eq:use-det-lemma2}) is based on the Weinstein-Aronszajn
identity, {\em i.e.}, $|\mathbf{I}_{m}+\mathbf{X}\mathbf{Y}|=|\mathbf{I}_{n}+\mathbf{Y}\mathbf{X}|$.

Taking the logarithm of (\ref{eq:use-det-lemma2}) and applying the
inverse property of the determinant, we further have
\begin{equation}
\log|\mathbf{Q}_{t}|-\log|\mathbf{Q}_{t|t-1}|=-\log|\mathbf{I}+\sigma_{\text{n}}^{-2}\mathbf{A}_{t}\mathbf{Q}_{t|t-1}\mathbf{A}_{t}^{\text{H}}|.\label{eq:difference-differential-entropy-2}
\end{equation}

Thus, Problem~(\ref{eq:problem-maximizing-entropy}) can be equivalently
rewritten as
\begin{align}
\mathop{\mbox{maximize}}\limits_{\mathbf{A}_{t}}\quad & |\mathbf{I}+\sigma_{\text{n}}^{-2}\mathbf{A}_{t}\mathbf{Q}_{t|t-1}\mathbf{A}_{t}^{\text{H}}|\label{eq:problem-maximizing-entropy-1}\\
\text{subject to}\quad & \mathbf{A}_{t}\mathbf{A}_{t}^{\text{H}}=\mathbf{I}.\nonumber 
\end{align}

Problem~(\ref{eq:problem-maximizing-entropy-1}) can be solved by
utilizing \ac{evd} and properties of the determinant of a matrix
as shown in the following proposition.
\begin{prop}[\foreignlanguage{american}{Adaptive sensing matrix}]
\label{prop:The-solution-to-At}The solution to problem~(\ref{eq:problem-maximizing-entropy-1})
is given by
\begin{equation}
{\bf A}_{t}=[{\bf w}_{t-1,1},\mathbf{w}_{t-1,2},\dots,\mathbf{w}_{t-1,m},\dots,\mathbf{w}_{t-1,M}]^{\mathrm{H}}\label{eq:optimal-solution-A}
\end{equation}
where ${\bf w}_{t-1,m}$ is the eigenvector corresponding to the $m$th
largest eigenvalue of the error covariance $\mathbf{Q}_{t|t-1}$.
\end{prop}
\begin{proof}
See Appendix~\ref{sec:Proof-of-Proposition-Solution-At}.
\end{proof}

Proposition~\ref{prop:The-solution-to-At} establishes a mechanism
for adaptively adjusting sensing matrices $\mathbf{A}_{t}$ by prioritizing
subspaces with high information gain. To handle cases where the $m$th
largest eigenvalue ($m\leq M$) exhibits multiplicity, resulting in
a subspace spanned by multiple eigenvectors, one may use a randomization
strategy to randomly select one eigenvector from this subspace.

In the initial stage, where prior information is scarce and $\mathbf{Q}_{t|t-1}$
approximates an identity matrix, the proposed strategy naturally degenerates
to employing random sensing vectors, which is aligned with the strategies
in the compressive sensing literature. As observations accumulate,
guided by (\ref{eq:differential-entropy-hat-ht}) and (\ref{eq:difference-differential-entropy-2}),
the strategy progressively refines $\mathbf{A}_{t}$ to minimize $\log|\mathbf{Q}_{t}|$,
thereby adaptively sweeping through the subspace $\mathbf{Q}_{t|t-1}$
of the channel and reducing the differential entropy $q(\hat{\mathbf{h}}_{t})$.
This aligns with the beam sweeping strategy in the literature.

The adaptive sensing matrices can be incorporated into the proposed
radio-map-embedded \ac{skf} framework for efficient channel sensing.
However, the design of $\mathbf{A}_{t}$ relies on the error covariance
$\mathbf{Q}_{t|t-1}$, which is based on the position distribution
$\pi_{t}(\mathbf{p}_{t})$ calculated using the measurement $\mathbf{y}_{t}$.
To address this implementation issue, the transition probability $\mathbb{P}(\mathbf{p}_{t}|\mathbf{p}_{t-1})$
is treated as the prior position distribution before observing $\mathbf{y}_{t}$.
Hence, the error covariance $\mathbf{Q}_{t|t-1}$ is first computed
using $\mathbb{P}(\mathbf{p}_{t}|\mathbf{p}_{t-1})$, and then refined
via $\pi_{t}(\mathbf{p}_{t})$ as detailed in Step~2a and Step~2c
of Algorithm~\ref{alg:radio-map-assisted-switching-Kalman-filter-csi-tracking},
respectively.

It is worth noting that the additional computational cost of the proposed
adaptive sensing approach is primarily due to the eigendecomposition
of $\mathbf{Q}_{t|t-1}$ since the design of $\mathbf{A}_{t}$ depends
on the error covariance $\mathbf{Q}_{t|t-1}$, whose computation has
already been accounted for in the CSI tracking complexity analysis
in Section~\ref{subsec:Switch-Kalman-Filtering}. While the eigendecomposition
for a full-rank matrix requires a complexity of $O(N_{t}^{3})$, there
exist efficient algorithms with a complexity of $O(N_{t}^{2})$ for
decomposing a low-rank matrix.

\section{Radio Map Construction}

\label{sec:Blind-Construction-of-Radio-Maps}

In this section, we construct the radio map~$\mathcal{M}$ via estimating
the hidden position variables $\mathbf{p}_{t}$ and the corresponding
channel covariances $\mathbf{C}(\mathbf{p}_{t})$ as formulated in
$\mathscr{P}_{2}$. Note that the construction discussed here is decoupled
from the \ac{csi} tracking problem studied in Section~\ref{sec:Radio-map-assisted-CSI-Tracking}.

While $\mathscr{P}_{2}$ also recovers the channel $\mathbf{h}_{t}$
and the position $\mathbf{p}_{t}$, it differs from the \ac{csi}
tracking problem in Section~\ref{sec:Radio-map-assisted-CSI-Tracking}
in the following aspects: First, as the ultimate goal of Section~\ref{sec:Radio-map-assisted-CSI-Tracking}
is to estimate the channel $\mathbf{h}_{t}$, the accuracy of the
hidden position $\mathbf{p}_{t}$ is not essential. By contrast, $\mathscr{P}_{2}$
aims at recovering accurate positions $\mathbf{p}_{t}$ for radio
map construction. Second, in $\mathscr{P}_{2}$, an off-line solution
is acceptable, where one can accumulate a large amount of sequential
data $\{\mathcal{Y}_{T}^{(k)}\}$ that probably comes from different
users $k$ from different routes to enhance the estimation of the
position $\mathbf{p}_{t}$. To illustrate the principle, we simply
consider there is one user in one route that accumulates a sufficient
amount of measurements. Third, while the anchor information for $\mathbf{p}_{t}$
comes from the radio map $\mathcal{M}$ in Section~\ref{sec:Radio-map-assisted-CSI-Tracking},
we need some side information $\bm{\Theta}$ to recover $\mathbf{p}_{t}$
in $\mathscr{P}_{2}$, because the radio map $\mathcal{M}$ is a variable
to be optimized.

In this paper, we assume that a very coarse location $\tilde{\mathbf{p}}_{t}$
can be obtained from the side information of the network topology
$\bm{\Theta}$, such as the \ac{bs} locations. There is active research
on coarse localization based on the sparse channel measurements $\mathbf{y}_{t}$
\cite{XinChe:J23,MagGioKanYu:J18}, such as using the \ac{wcl} algorithm.
It has been demonstrated in \cite{XinChe:J23} using real data that
a localization error of $22$ meters can be achieved by measuring
only the \ac{rss} from the \acpl{bs}. In general, a $20\sim60$
meter \ac{rss}-based localization error is reported in the literature
\cite{MagGioKanYu:J18}.

While a localization error of $20$ meters is still too rough for
radio map construction, the fundamental question to be investigated
in this section is: {\em Can the radio-map-embedded  CSI model substantially enhance the localization performance from coarse locations $\tilde{\mathbf{p}}_t$?}

\subsection{Memory-Assisted Trajectory Discovery}

\label{subsec:Problem-for-Radio-Map-Construction-Channel-Est}

Assume that a very coarse position $\tilde{\mathbf{p}}_{t}$ is available
at time slot $t$. Denoting $\tilde{\mathcal{P}}_{T}=(\tilde{\mathbf{p}}_{1},\tilde{\mathbf{p}}_{2},\dots,\tilde{\mathbf{p}}_{T})$,
we define an $L_{2}$-norm value function as $f(\mathcal{Y}_{T},\mathcal{P}_{T};\bm{\Theta})=\sum_{t=1}^{T}\|\mathbf{p}_{t}-\tilde{\mathbf{p}}_{t}\|_{2}$.
Note that how to fuse the coarse position $\tilde{\mathbf{p}}_{t}$
to $\mathscr{P}_{2}$ is not the focus here, and other specifications
of the value function $f(\mathcal{Y}_{T},\mathcal{P}_{T};\bm{\Theta})$
may also work. Problem~$\mathscr{P}_{2}$ can be rewritten as
\begin{align*}
\mathscr{P}_{2}':\qquad\mathop{\mbox{maximize}}\limits_{\mathcal{H}_{T},\mathcal{P}_{T},\mathcal{M}} & \quad\log p(\mathcal{Y}_{T},\mathcal{H}_{T},\mathcal{P}_{T};\mathcal{M})\\
 & \quad\quad\:-\mu\sum_{t=1}^{T}\|\mathbf{p}_{t}-\tilde{\mathbf{p}}_{t}\|_{2}\\
\text{subject to} & \quad\mathbf{p}_{t}\in\mathcal{X},\qquad t=1,2,\dots,T
\end{align*}
where $\mu>0$ serves as a regularization parameter.

Since the observation $\mathbf{y}_{t}$ depends on the hidden user
position $\mathbf{p}_{t}$ and the transition probability of $\mathbf{p}_{t}$
is depicted by a mobility model~(\ref{eq:mobility-transition-probability-1}),
we can formulate an \ac{hmm} problem with a regularization term for
tracking the sequential states evolution of $\mathcal{Y}_{T}$ and
$\mathcal{P}_{T}$. Specifically, from the observation model~(\ref{eq:observation_model}),
since $\mathbf{y}_{t}$ depends on $\mathbf{h}_{t}$ which depends
on $\mathbf{p}_{t}$, we can marginalize $\mathbf{h}_{t}$ from $p(\mathcal{Y}_{T},\mathcal{H}_{T},\mathcal{P}_{T};\mathcal{M})$
to obtain
\begin{align}
p(\mathcal{Y}_{T},\mathcal{P}_{T};\mathcal{M}) & =\prod_{t=1}^{T}p(\mathbf{y}_{t}|\mathbf{p}_{t};\mathcal{M})\times\prod_{t=2}^{T}\mathbb{P}(\mathbf{p}_{t}|\mathbf{p}_{t-1})\mathbb{P}(\mathbf{p}_{1})\label{eq:probability-YT-PT}
\end{align}
which follows the same derivation as that in~(\ref{eq:pi-pt}).

Consequently, based on Problem~$\mathscr{P}_{2}'$, an \ac{hmm}
problem with a regularization term is formulated as
\begin{align}
\mathop{\mbox{maximize}}\limits_{\mathcal{P}_{T}} & \quad\sum_{t=1}^{T}(\log p(\mathbf{y}_{\text{\ensuremath{t}}}|\mathbf{p}_{t};\mathcal{M})+\mu\|\mathbf{p}_{t}-\tilde{\mathbf{p}}_{t}\|_{2})\nonumber \\
 & \quad\quad\:\,+\sum_{t=2}^{T}\log\mathbb{P}(\mathbf{p}_{\text{\ensuremath{t}}}|\mathbf{p}_{t-1})+\log\mathbb{P}(\mathbf{p}_{1})\label{eq:localization-problem}\\
\text{subject to} & \quad\mathbf{p}_{t}\in\mathcal{X},\qquad t=1,2,\dots,T.\nonumber 
\end{align}

As indicated in (\ref{eq:localization-problem}), the trajectory discovery
problem can be framed as an \ac{hmm} decoding task, where the objective
is to find the most likely sequence of hidden states (locations) that
maximizes the regularized log-likelihood. Specifically, the log-likelihood
$\log p(\mathbf{y}_{\text{\ensuremath{t}}}|\mathbf{p}_{t};\mathcal{M})+\log\mathbb{P}(\mathbf{p}_{\text{\ensuremath{t}}}|\mathbf{p}_{t-1})$
captures the hidden dynamics of the trajectory $\mathcal{P}_{T}$,
while the regularization term $\mu\|\mathbf{p}_{t}-\tilde{\mathbf{p}}_{t}\|_{2}$
is to calibrate the topology $\mathbf{p}_{t}$ with the geographic
side information $\tilde{\mathbf{p}}_{t}$ from the physical world.

To solve Problem~(\ref{eq:localization-problem}), the Viterbi algorithm
can be employed to traverse all possible location sequences using
dynamic programming to compute the optimal trajectory efficiently
\cite{PreTeuVetFla:B07}. At each time step, the algorithm recursively
evaluates the combined likelihood and transition probabilities of
the best location sequence to a state $\mathbf{p}_{t}$, while incorporating
the regularization term $\mu\|\mathbf{p}_{t}-\tilde{\mathbf{p}}_{t}\|_{2}$
to penalize deviations from $\tilde{\mathbf{p}}_{t}$. By maintaining
backtracking pointers at each step, the algorithm tracks the location
sequence with the highest accumulated objective value. This process
ensures implicit exploration of all possible trajectories, with the
most likely sequence retained as the final solution to Problem~(\ref{eq:localization-problem}).

The time complexity for solving Problem~(\ref{eq:localization-problem})
using the Viterbi algorithm is $O(T|\mathcal{X}|^{2}N_{\text{t}}^{3})$,
where $O(T|\mathcal{X}|^{2})$ is the time complexity of the Viterbi
algorithm to handle a scalar variable transition process and $O(N_{\text{t}}^{3})$
comes from the matrix inverse operation in $p(\mathbf{y}_{\text{\ensuremath{t}}}|\mathbf{p}_{t})$
at each time step.

\subsection{Covariance Estimation from Sparse \ac{csi} Measurements}

\label{subsec:Channel-Coveriance-Construction}

Denote $\mathcal{T}_{i}\triangleq\{t:\mathbf{p}_{t}=\mathbf{x}_{i}\in\mathcal{X}\}$
as the set of measurements taken in the $i$th grid cell. We need
to construct an estimator $\ensuremath{\hat{\mathbf{C}}}(\mathbf{x}_{i})$
for the covariance matrix $\mathbf{C}(\mathbf{x}_{i})=\mathbb{E}\{\mathbf{h}_{t}\mathbf{h}_{t}^{\mathrm{H}}\}\in\mathbb{C}^{N_{\text{t}}\times N_{\text{t}}}$
based on the measurements $\mathbf{y}_{t}\in\mathbb{C}^{M}$, where
$t\in\mathcal{T}_{i}$ and $M\ll N_{\text{t}}$. Assume that the measurements
$\mathbf{y}_{t}$ in the $i$th grid cell are independent. Otherwise,
to ensure independence, a random subset of data can be sampled from
the collected dataset. While we assume that there are a sufficient
number of measurements that have been collected for discretized locations
$\mathbf{x}_{i}$, it is still a challenge that the measurement $\mathbf{y}_{t}$
has a very low dimension $M$.

\subsubsection{Unbiased Estimator of the Covariance}

\label{subsec:Unbiased-Estimator-of}

The general idea is to construct a data point in the $N_{\text{t}}$-dimensional
subspace spanned by the sensing matrix $\mathbf{A}_{t}$, and thus,
the data point is defined as
\begin{equation}
\boldsymbol{\varphi}_{t}=\mathbf{A}_{t}^{\mathrm{H}}\mathbf{y}_{t}=\mathbf{A}_{t}^{\mathrm{H}}\mathbf{A}_{t}\mathbf{h}_{t}+\mathbf{A}_{t}^{\text{H}}\mathbf{n}_{t}=\boldsymbol{\Phi}_{t}\mathbf{h}_{t}+\mathbf{A}_{t}^{\text{H}}\mathbf{n}_{t}
\end{equation}
where $\boldsymbol{\Phi}_{t}=\mathbf{A}_{t}^{\text{H}}\mathbf{A}_{t}\in\mathbb{C}^{N_{\mathrm{t}}\times N_{\mathrm{t}}}$.

We then construct a sample covariance matrix as follows
\begin{align}
\hat{\boldsymbol{\Omega}}_{\text{y}}(\mathbf{x}_{i}) & \triangleq\frac{N_{\mathrm{t}}^{2}}{M^{2}}\frac{1}{|\mathcal{T}_{i}|}\sum_{t\in\mathcal{T}_{i}}\boldsymbol{\varphi}_{t}\boldsymbol{\varphi}_{t}^{\mathrm{H}}\nonumber \\
 & =\frac{N_{\mathrm{t}}^{2}}{|\mathcal{T}_{i}|M^{2}}\sum_{t\in\mathcal{T}_{i}}(\boldsymbol{\Phi}_{t}\mathbf{h}_{t}+\mathbf{A}_{t}^{\mathrm{H}}\mathbf{n}_{t})(\boldsymbol{\Phi}_{t}\mathbf{h}_{t}+\mathbf{A}_{t}^{\mathrm{H}}\mathbf{n}_{t})^{\text{H}}\label{eq:def-rescaled-version-of-observed-covariance}
\end{align}

It can be shown that an unbiased estimator for the covariance matrix
$\mathbf{C}(\mathbf{x}_{i})$ can be constructed from the sample covariance
matrix $\hat{\boldsymbol{\Omega}}_{\text{y}}(\mathbf{x}_{i})$ as
shown in the following proposition.
\begin{prop}[Unbiased estimator]
\label{prop:estimator-of-sample-covariance}Let $\mathbf{A}_{t}\mathbb{\in\mathbb{C}}^{M\times N_{\text{t}}}$
be an orthonormal basis for an $M$-dimensional subspace drawn uniformly
at random. An unbiased estimator for $\mathbf{C}(\mathbf{x}_{i})$
is given by
\begin{align}
\text{\ensuremath{\hat{\mathbf{C}}}}(\mathbf{x}_{i})= & \frac{M(\bar{N}_{\mathrm{t}}\hat{\boldsymbol{\Omega}}_{\text{y}}(\mathbf{x}_{i})-(N_{\mathrm{t}}-M)\mathrm{tr}(\hat{\boldsymbol{\Omega}}_{\text{y}}(\mathbf{x}_{i}))\mathbf{I})}{N_{\mathrm{t}}(N_{\mathrm{t}}M+N_{\mathrm{t}}-2)}\nonumber \\
 & -\frac{\sigma_{\mathrm{n}}^{2}(M\bar{N}_{\mathrm{t}}\boldsymbol{\Omega}_{\mathrm{A}}(\mathbf{x}_{i})-N_{\mathrm{t}}(N_{\mathrm{t}}-M)\mathrm{tr}(\boldsymbol{\Omega}_{\mathrm{A}}(\mathbf{x}_{i}))\mathbf{I})}{M(N_{\mathrm{t}}M+N_{\mathrm{t}}-2)}\label{eq:unbiased-estimate-sampling-covariance}
\end{align}
where $\bar{N}_{\mathrm{t}}\triangleq(N_{\mathrm{t}}+2)(N_{\mathrm{t}}-1)$
and $\boldsymbol{\Omega}_{\mathrm{A}}(\mathbf{x}_{i})\triangleq\frac{1}{|\mathcal{T}_{i}|}\sum_{t\in\mathcal{T}_{i}}\mathbf{A}_{t}^{\mathrm{H}}\mathbf{A}_{t}$.
\end{prop}
\begin{proof}
See Appendix~\ref{sec:Proof-of-Proposition-estimator-C}.
\end{proof}

The unbiased estimator $\text{\ensuremath{\hat{\mathbf{C}}}}(\mathbf{x}_{i})$
in Proposition~\ref{prop:estimator-of-sample-covariance} is composed
of a scaled version of the sample covariance of the \ac{csi} measurements,
with an additional correction term that accounts for the noise covariance.
Consider the special case where $M=N_{\text{t}}$ and $\boldsymbol{\Phi}_{t}=\mathbf{A}_{t}^{\text{H}}\mathbf{A}_{t}=\mathbf{I}$,
the estimator simplifies to $\ensuremath{\hat{\mathbf{C}}}(\mathbf{x}_{i})=\sum_{t\in\mathcal{T}_{i}}(\boldsymbol{\varphi}_{t}\boldsymbol{\varphi}_{t}^{\mathrm{H}}-\sigma_{\mathrm{n}}^{2}\mathbf{I})/|\mathcal{T}_{i}|$.
Next, we can take the expectation of $\text{\ensuremath{\hat{\mathbf{C}}}}(\mathbf{x}_{i})$
and obtain that $\mathbb{E}\{\ensuremath{\hat{\mathbf{C}}}(\mathbf{x}_{i})\}=\mathbb{E}\{\sum_{t\in\mathcal{T}_{i}}((\mathbf{h}_{t}+\mathbf{A}_{t}^{\text{H}}\mathbf{n}_{t})(\mathbf{h}_{t}+\mathbf{A}_{t}^{\text{H}}\mathbf{n}_{t})^{\mathrm{H}}-\sigma_{\mathrm{n}}^{2}\mathbf{I})/|\mathcal{T}_{i}|\}=\mathbf{C}(\mathbf{x}_{i})$.

Additionally, it is observed from (\ref{eq:unbiased-estimate-sampling-covariance})
that the environmental noise is rescaled by $\mathbf{A}_{t}$ and
subsequently removed. In high \ac{snr} scenarios where $\sigma_{\text{n}}^{2}$
is small, the first term (measurement term) in (\ref{eq:unbiased-estimate-sampling-covariance})
dominates the channel covariance estimation while the second term
(noise term) can be considered negligible.

\subsubsection{Upper Bound of Reconstruction Error}

\label{subsec:Estimation-Error-Analysis}

Here, we present the reconstruction error of the channel covariance
$\mathbf{C}(\mathbf{x}_{i})$ \ac{wrt} the \ac{csi} dimension $N_{\mathrm{t}}$,
the observation dimension $M$, the number of observations $|\mathcal{T}_{i}|$,
and the rank of the true channel covariance $\mathbf{C}(\mathbf{x}_{i})$.

Consider the high \ac{snr} case, where $\sigma_{\text{n}}=0$, and
the observed \ac{csi} is simplified as $\mathbf{y}_{t}=\mathbf{A}_{t}\mathbf{h}_{t}$.
Based on (\ref{eq:unbiased-estimate-sampling-covariance}), the unbiased
estimator $\text{\ensuremath{\hat{\mathbf{C}}}}(\mathbf{x}_{i})$
in (\ref{eq:unbiased-estimate-sampling-covariance}) simplifies to
\begin{equation}
\ensuremath{\hat{\mathbf{C}}'}(\mathbf{x}_{i})=\frac{M(\bar{N}_{\mathrm{t}}\hat{\boldsymbol{\Omega}}_{\text{y}}(\mathbf{x}_{i})-(N_{\mathrm{t}}-M)\mathrm{tr}(\hat{\boldsymbol{\Omega}}_{\text{y}}(\mathbf{x}_{i}))\mathbf{I})}{N_{\mathrm{t}}(N_{\mathrm{t}}M+N_{2}-2)}.\label{eq:estimator-without-noise}
\end{equation}

The following proposition shows the spectral upper bound of the reconstruction
error $\ensuremath{\hat{\mathbf{C}}'}(\mathbf{x}_{i})-\mathbf{C}(\mathbf{x}_{i})$.
\begin{prop}[Upper bound of the reconstruction error]
\label{prop:upper-bound-of-the-estimation-error}Assuming $\mathrm{rank}(\mathbf{C}(\mathbf{x}_{i}))\leq R$,
there exists a universal constant $\kappa>0$ such that for any $\zeta\in(0,1)$,
when $N_{\mathrm{t}}\geq2$ and $|\mathcal{T}_{i}|\geq N_{\mathrm{t}}\log(1/\zeta)$,
with probability at least $1-\zeta$,
\begin{align}
\|\ensuremath{\hat{\mathbf{C}}'}(\mathbf{x}_{i})-\mathbf{C}(\mathbf{x}_{i})\|_{2} & \leq\frac{\kappa}{\sqrt{|\mathcal{T}_{i}|}}\|\mathbf{C}(\mathbf{x}_{i})\|_{2}\left(S_{1}+S_{2}+S_{3}\right)\label{eq:upper-bound-estimation-error-1}
\end{align}
where $S_{1}=\frac{\sqrt{N_{\mathrm{t}}R^{2}\log^{2}(|\mathcal{T}_{i}|N_{\mathrm{t}}/\zeta)}}{M}$,
$S_{2}=\sqrt{R\log(1/\zeta)}$, and $S_{3}=\frac{N_{\mathrm{t}}R\log^{2}(|\mathcal{T}_{i}|N_{\mathrm{t}}/\zeta)}{\sqrt{|\mathcal{T}_{i}|}M}$.
\end{prop}
\begin{proof}
Since ${\bf h}_{t}$ follows a Gaussian distribution $\mathcal{CN}(0,{\bf C}(\mathbf{x}_{i}))$
and $\ensuremath{\hat{\mathbf{C}}'}(\mathbf{x}_{i})$ is an unbiased
estimator given in (\ref{eq:estimator-without-noise}), applying Corollary
1 (Gaussian Upper Bounds) and Corollary 2 in \cite{AziKriSin:J18}
to our case automatically yields the results.
\end{proof}

Proposition~\ref{prop:upper-bound-of-the-estimation-error} provides
key insights into the factors influencing the reconstruction error
of sample covariance. First, if $|\mathcal{T}_{i}|$ is large relative
to $N_{\mathrm{t}}^{2}/M^{2}$ and the logarithmic factors are ignored,
the leading terms in the error bound~(\ref{eq:upper-bound-estimation-error-1})
is $\tilde{\mathcal{O}}(\|\mathbf{C}(\mathbf{x}_{i})\|_{2}(\sqrt{N_{\mathrm{t}}R^{2}/(|\mathcal{T}_{i}|M^{2})}+N_{\mathrm{t}}R/(|\mathcal{T}_{i}|M))),$
where the notation $\tilde{\mathcal{O}}(n)$ represents the asymptotic
growth rate of a function $n$ while suppressing dependencies on logarithmic
factors. Thus, increasing the number of observation samples reduces
error terms inversely, making it crucial to accumulate sufficient
samples. Second, the upper bound in~(\ref{eq:upper-bound-estimation-error-1})
decreases inversely with $M$, and thus, increasing $M$ may result
in a reduction of the reconstruction error. Additionally, the rank
of the true channel covariance highlights the difficulty of estimating
channel covariance in \ac{nlos} environments with rich multipath
effects. By contrast, in \ac{los} environments, where $R$ is small,
reconstruction accuracy improves naturally, requiring fewer pilot
observations (see Table~\foreignlanguage{american}{\ref{tab:NMSE-C}}).

\subsection{An Iterative Framework for Radio Map Construction}

\label{subsec:Joint-Radio-Map-Trajectory-Discovery-CSI-Tracking}

To sum up, a radio map construction framework is developed as follows.
We begin with a sequence of sparse \ac{csi} measurements and apply
the Viterbi algorithm to solve the user trajectory discovery problem~(\ref{eq:localization-problem}).
Based on the resulting user position estimates, the sparse \ac{csi}
measurements are then assigned to the corresponding grid cells. Subsequently,
we employ the unbiased estimator $\ensuremath{\hat{\mathbf{C}}}(\mathbf{x}_{i})$
as defined in (\ref{eq:unbiased-estimate-sampling-covariance}) to
reconstruct the channel covariance $\mathbf{C}(\mathbf{x}_{i})$ in
each grid cell. This process yields the radio map~$\mathcal{M}$,
which is then used to perform a second round of trajectory discovery.
Such a process is repeated until the change in estimated user positions
between two iterations falls below a predefined threshold $\epsilon$,
at which point the radio map~$\mathcal{M}$ is considered finalized.
Note that the trajectory discovery problem~(\ref{eq:localization-problem})
requires an initialization of the channel covariance ${\bf C}(\mathbf{p}_{t})$
to calculate $p(\mathbf{y}_{\text{\ensuremath{t}}}|\mathbf{p}_{t})$.
One can simply initialize identical channel covariance matrices for
each discretized grid cell since $p(\mathbf{y}_{\text{\ensuremath{t}}}|\mathbf{p}_{t})$
is only utilized for capturing relative relationships of user positions.
The technical details are shown in Algorithm~\ref{alg:radio-map-construction}.

The iterative framework for radio map construction converges naturally
as the trajectory and covariance estimates are progressively refined.
In each iteration, the updated trajectory helps allocate \ac{csi}
measurements to grid cells, improving the accuracy of the channel
covariance matrices. These refined covariances then aid in discovering
more accurate trajectories in the next iteration. This feedback loop
stabilizes as changes in the estimated user positions decrease, eventually
meeting the predefined threshold $\epsilon$. While a formal proof
of convergence is beyond this discussion, numerical results in Section~\ref{sec:Numerical-Results}
confirm the practical convergence of the framework (see Fig.~\ref{fig:convergence-alg2}).

\begin{algorithm}
\textbf{Input:} The sparse \ac{csi} measurements $\mathcal{Y}_{T}$,
coarse positions $\tilde{\mathcal{P}}_{T}$, the noise variance $\sigma_{\text{n}}^{2}$,
and a threshold $\epsilon$.
\begin{enumerate}
\item Initialize $\mathcal{M}^{(l)}=\{(\mathbf{x},\mathbf{I}):\mathbf{x}\in\mathcal{X}\}$
and $\hat{\mathcal{P}}_{T}^{(l)}=\tilde{\mathcal{P}}_{T}$ at the
iteration $l=0$.
\item Repeat the following process until $\frac{1}{T}\sum_{t=1}^{T}\|\mathbf{p}_{t}^{(l)}-\mathbf{p}_{t}^{(l-1)}\|_{2}<\epsilon$:
\begin{enumerate}
\item Set $l\leftarrow l+1$.
\item Estimate the user trajectory $\hat{\mathcal{P}}_{T}^{(l)}=\{\hat{\mathbf{p}}_{1}^{(l)},\hat{\mathbf{p}}_{2}^{(l)},\dots,\hat{\mathbf{p}}_{T}^{(l)}\}$
by solving Problem~(\ref{eq:localization-problem}).
\item Estimate the channel covariance $\hat{\mathbf{C}}^{(l)}(\mathbf{x}_{i})$
using $\mathcal{Y}_{T}$, $\hat{\mathcal{P}}_{T}^{(l)}$, and (\ref{eq:unbiased-estimate-sampling-covariance}),
and construct the radio map $\hat{\mathcal{M}}^{(l)}=\{(\mathbf{x},\hat{\mathbf{C}}^{(l)}({\bf x})):\mathbf{x}\in\mathcal{X}\}$.
\end{enumerate}
\item Output $\hat{\mathcal{M}}^{(l)}$ and $\hat{\mathcal{P}}_{T}^{(l)}$.
\end{enumerate}
\caption{Radio Map Construction}

\label{alg:radio-map-construction}
\end{algorithm}

\section{Numerical Results}

\label{sec:Numerical-Results}

In this section, we present experimental findings on real-world city
maps with 3D ray-traced \ac{mimo} channel datasets.

\subsection{Environment Setup and Scenarios}

\label{subsec:Environment-Setup-and-Scenarios}

The experiments are conducted on a {\ac{mimo} \ac{csi} database
generated by Wireless Insite \cite{remcom2025} which provides 3D
ray-tracing for the analysis of site-specific radio wave propagation.
The tested city topology is a $710$-meter $\times$ $740$-meter
area from San Francisco, America. Seven \acpl{bs} are deployed on
the top of some buildings as shown in Fig.~\ref{fig:simulation_environments}(a).
The \ac{mimo} antenna of each \ac{bs} consists of $N_{\mathrm{t}}=64$
antenna arrays as shown in Fig.~\ref{fig:simulation_environments}(b).
By simulating random walks along roads in the map, the radio database
is established, and it contains $19331$ user locations with $19331\times7$
\ac{mimo} channels.

The tested user trajectories are randomly generated along the road.
The measurement ${\bf y}_{t}^{(q)}$ for the $q$th \ac{bs} at the
position ${\bf p}_{t}$ is generated using (\ref{eq:observation_model}),
where the noise variance $\sigma_{\text{n}}^{2}$ is generated using
the mean power of channels and the given \ac{snr}, {\em i.e.}, $\sigma_{\mathrm{n}}^{2}=\mathbb{E}\{\|\mathbf{h}_{t}^{(q)}\|_{2}^{2}\}/\text{10}^{\mathrm{SNR}/10}$.
Coarse user positions $\tilde{\mathbf{p}}_{t}$ are generated by adding
a Gaussian random shift $\mathcal{N}(0,900)$ to the true position
${\bf p}_{t}$. The radio map grid has a resolution of $5$~meters,
with the user moving at $10$~meters/second and observations sampled
every $10$~milliseconds. The threshold $\epsilon$ for position
change in Algorithm~\ref{alg:radio-map-construction} is set to $0.5$~meter.
The dimension~$M$ of ${\bf y}_{t}^{(q)}$ is set to $1$ for the
proposed scheme.

Note that Doppler shift considerations are not necessary in our experimental
setup. This is because the primary purpose of \ac{csi} estimation
in the proposed scheme is to enable \ac{mimo} beamforming at the
\ac{bs}. Specifically, it aims to find a beamforming vector ${\bf b}$
that maximizes the signal strength $|{\bf b}^{\text{H}}{\bf h}e^{-j\theta}|$,
where the phase offset $\theta$ due to the Doppler effect does not
affect the signal strength. During beamformed data transmission, additional
pilot observations are still required to compensate for the phase
offset $\theta$. However, this aspect is beyond the scope of this
work, which centers on reducing pilot overhead for \ac{csi} tracking
to enable effective beamforming in mobile scenarios.

\begin{figure}
\begin{centering}
\subfigure[]{\includegraphics[width=0.65\columnwidth]{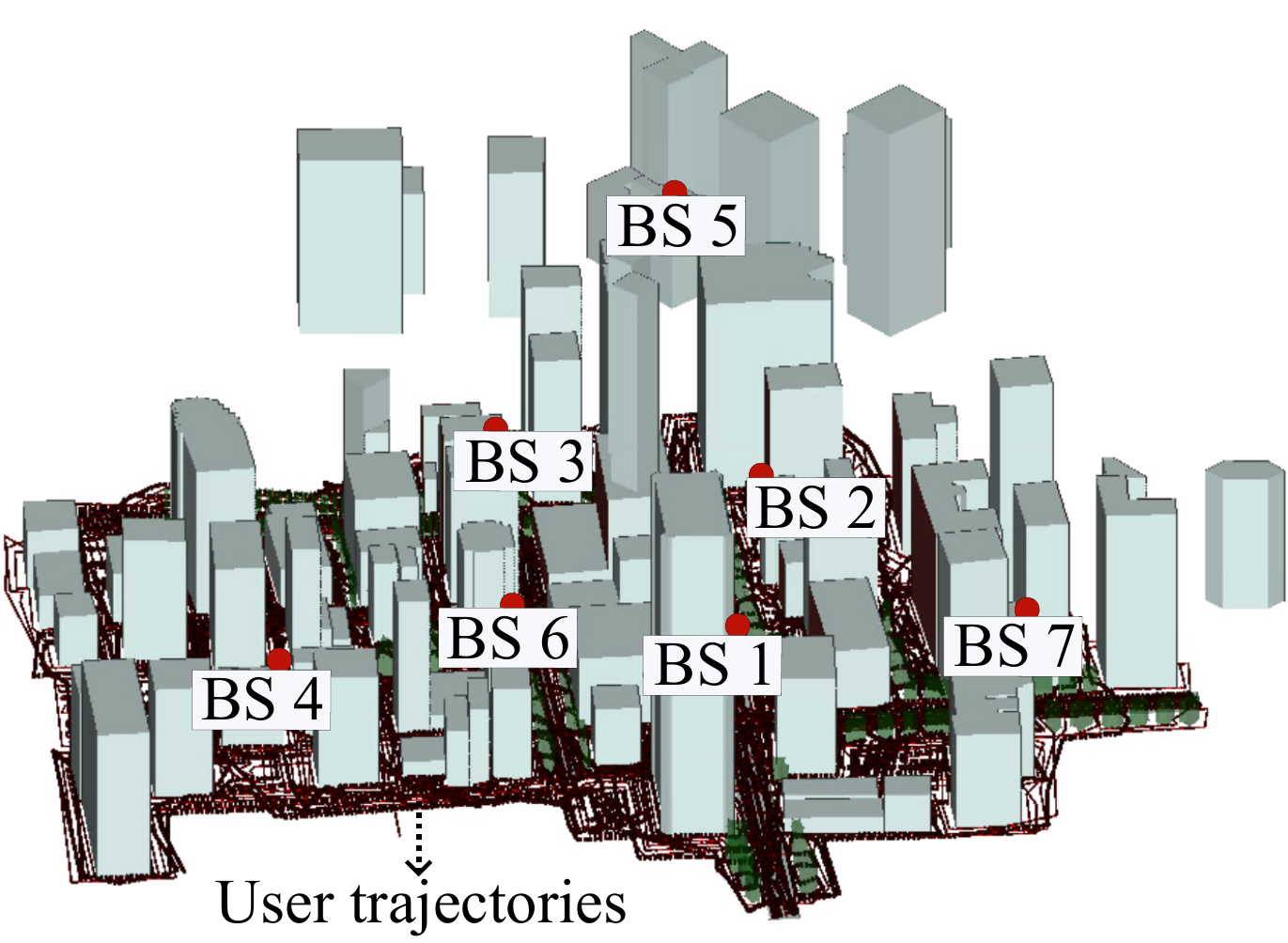}}\subfigure[]{\includegraphics[width=0.3\columnwidth]{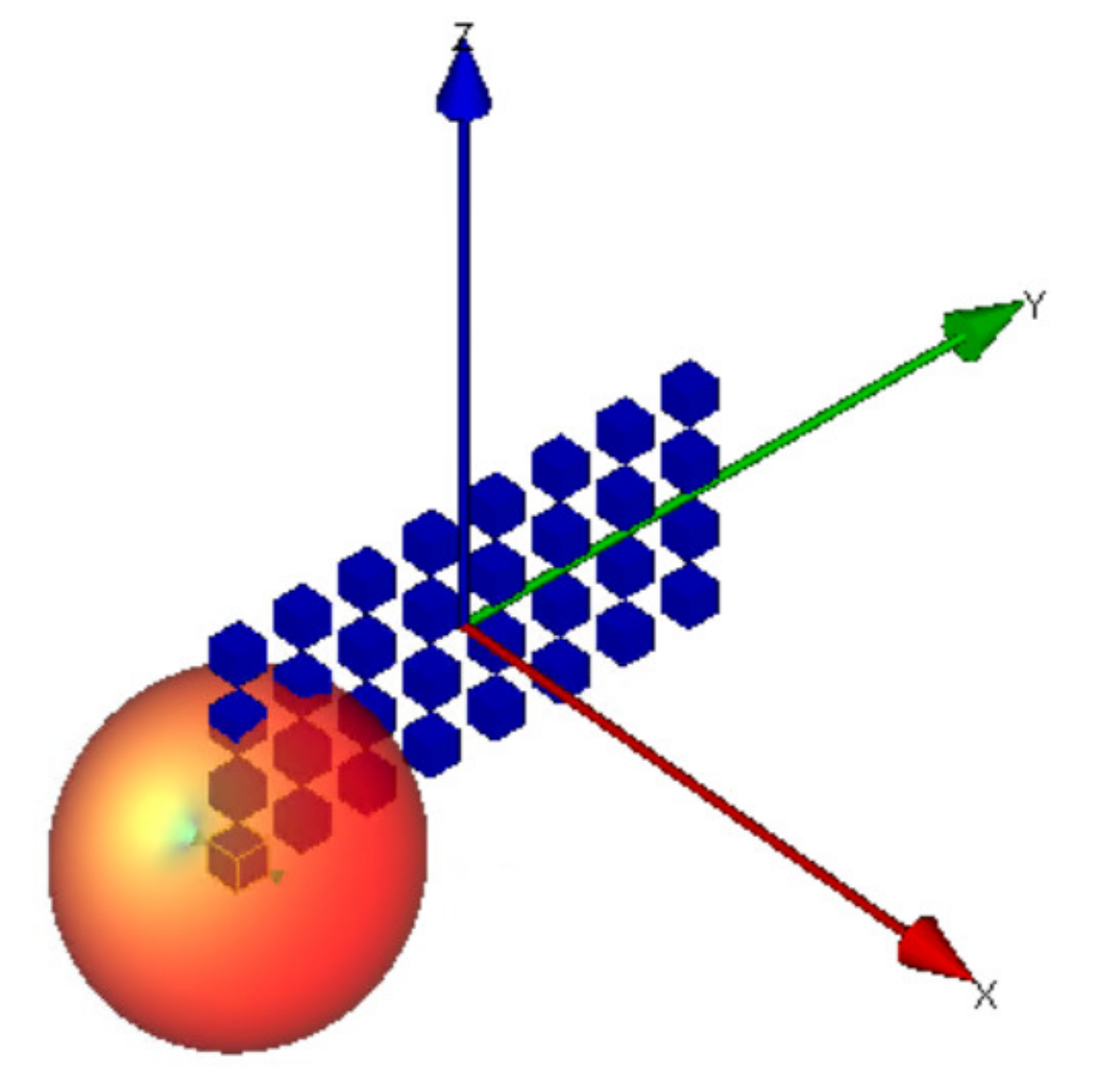}}
\par\end{centering}
\caption{(a) Simulation environment with $7$ \acpl{bs}\foreignlanguage{american}{;
(b) \Ac{mimo} antenna with $32$ dual-polarized dipole antenna arrays
($N_{\mathrm{t}}=64$).}}

\label{fig:simulation_environments}
\end{figure}

\selectlanguage{american}%
\begin{figure*}
\begin{centering}
\subfigure[]{\includegraphics[width=0.66\columnwidth]{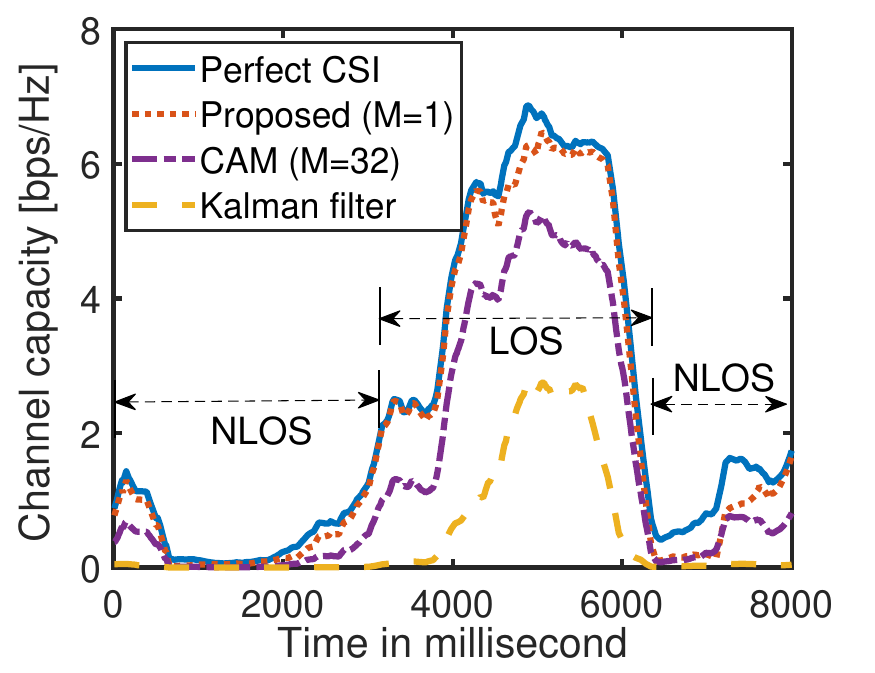}}\subfigure[]{\includegraphics[width=0.66\columnwidth]{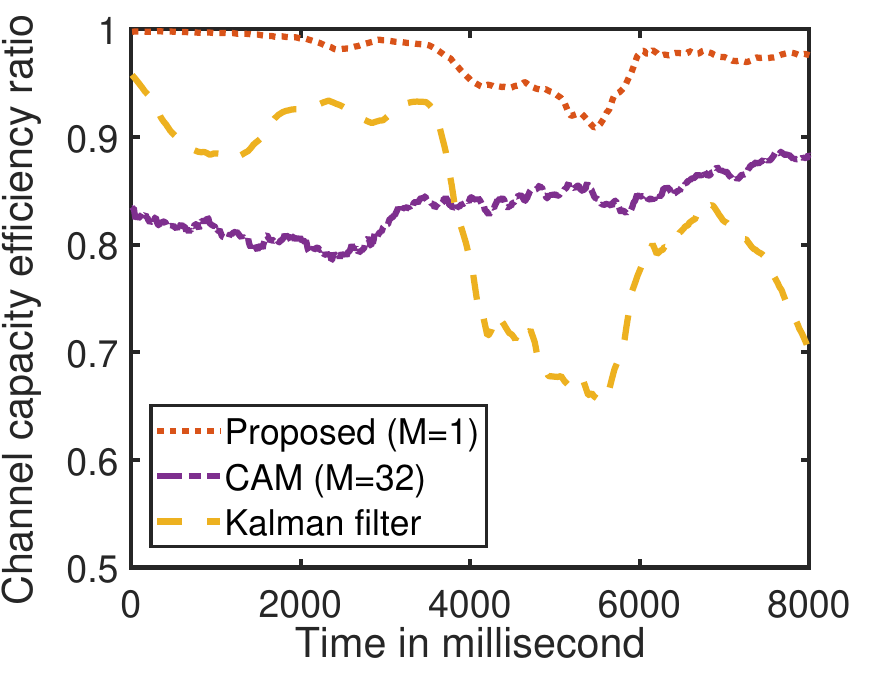}}\subfigure[]{\includegraphics[width=0.66\columnwidth]{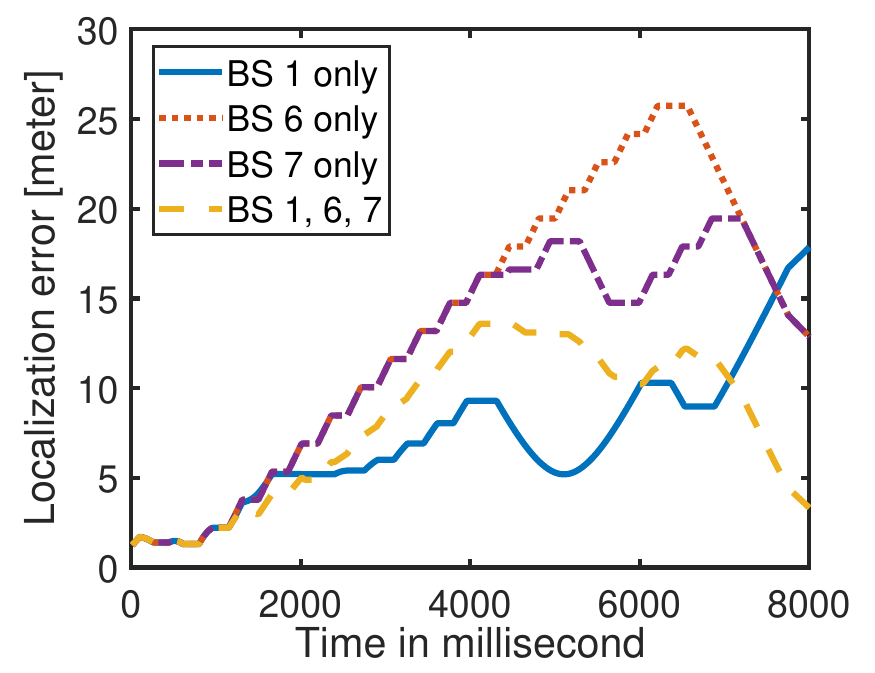}}
\par\end{centering}
\caption{\Ac{csi} and user position tracking performance. (a) \foreignlanguage{english}{Channel
capacity for a trajectory transitioning through \ac{nlos}, \ac{los},
and \ac{nlos} propagation conditions to BS 7; (b) Channel capacity
efficiency ratio for an \ac{los} trajectory to BS 1; (c) Localization
error of the proposed radio-map-embedded CSI tracking scheme for the
same trajectory in (a), which is \ac{los} to BS 1 and \ac{nlos}
to BS 6.}}

\label{fig:user_movement_process}
\end{figure*}

\selectlanguage{english}%

\subsection{Baseline Schemes}

The following baseline schemes are compared with the proposed \ac{csi}
tracking scheme.
\begin{itemize}
\item Genius-aided \ac{csi} estimation: This scheme performs the proposed
radio-map-embedded \ac{csi} tracking algorithm with known user positions
and perfect radio maps.
\item \Ac{cam}-aided \ac{csi} estimation with one-shot measurement \cite{WuZen:C23}:
This scheme constructs a \ac{cam} to map locations to path gain magnitudes
at known angles. The \ac{cam} serves as a localization fingerprint.
$M=N_{t}/2$ beam vectors are designed based on the estimated \ac{aod}
from the \ac{cam} for channel sensing, with the \ac{ls} method for
channel estimation.
\item \Ac{ar}-based \ac{csi} prediction \cite{VinJunHam:C24}: This approach
estimates \ac{ar} coefficients of the channel based on the Yule-Walker
equations and estimates the channel ${\bf h}_{t}$ using $10$ past
channel estimates. The \ac{ls} method is used for initial channel
estimation with randomly generated $\mathbf{A}_{t}\in\mathbb{C}^{N_{\mathrm{t}}\times N_{\mathrm{t}}}$.
\item \Ac{lstm}-based \ac{csi} prediction \cite{PenZhaCheYan:C20}: This
method adopts a recurrent neural network (RNN) and uses the previous
$10$ estimated channels to predict ${\bf h}_{t}$. Other settings
are the same as those of the \ac{ar}-based method.
\item \Ac{kf}-based \ac{csi} tracking \cite{WelBis:B95}: This scheme
performs \ac{csi} tracking using \ac{kf} with $M=1$, and no radio
map is used.
\item \Ac{ls}-based \ac{csi} estimation: An \ac{ls} estimate is given
by $\hat{{\bf h}}={\bf A}_{t}^{\dagger}{\bf y}$. Note that since
we typically have $M<N_{t}$, we need to perform pseudo-inverse for
the matrix $\mathbf{A}_{t}$. The sensing matrix $\mathbf{A}_{t}$
is randomly generated with $M=N_{t}/2$.
\end{itemize}

The following metrics are used to evaluate the proposed algorithms:
1) Localization error: Given the true user position $\mathbf{p}_{t}$
and the estimated user position $\hat{\mathbf{p}}$, the localization
error is defined as $\frac{1}{T}\sum_{t=1}^{T}\|{\bf p}_{t}-\hat{{\bf p}}_{t}\|_{2}.$
2) Channel capacity: Using the \ac{mrc} scheme, the channel capacity
is expressed as $f(\mathbf{b},{\bf h})=\log_{2}(1+|\mathbf{b}^{\text{H}}{\bf h}|^{2}/\sigma_{\text{n}}^{2}),$
where ${\bf h}$ is the true \ac{csi} and $\mathbf{b}=\hat{\mathbf{h}}/\|\hat{\mathbf{h}}\|_{2}$
is the beamforming vector, in which $\hat{\mathbf{h}}$ is the estimated
\ac{csi}. 3) Channel capacity efficiency ratio: This is defined as
the ratio of the maximum capacity achieved with $\hat{\mathbf{h}}$
relative to that with perfect \ac{csi} ${\bf h}$. 4) Normalized
$L_{2}$-norm estimation error of channel covariance: This is defined
as $\frac{1}{|\mathcal{X}|}\sum_{\mathbf{x}\in\mathcal{X}}\|\hat{\mathbf{C}}(\mathbf{x})-\mathbf{C}(\mathbf{x})\|_{2}/\|\mathbf{C}(\mathbf{x})\|_{2}$.
5) Normalized projection error of channel covariance: This is defined
as $\|\mathbf{C}-\hat{\mathbf{C}}(\hat{\mathbf{C}}^{\mathrm{H}}\hat{\mathbf{C}})^{-1}\hat{\mathbf{C}}^{\mathrm{H}}\mathbf{C}\|_{\mathrm{F}}/\|\mathbf{C}\|_{\mathrm{F}}$,
which quantifies the mismatch between the true and estimated signal
subspaces.

\subsection{Radio-Map-Embedded CSI Tracking}

\label{subsec:Radio-Map-assisted-Trajectory}

We first evaluate the proposed radio-map-embedded \ac{csi} tracking
scheme. A perfect radio map~$\mathcal{M}$ for each \ac{bs} is constructed
by discretizing the area of interest and calculating $\mathbf{C}(\mathbf{x}_{i})$
using true channels in each grid cell centered at $\mathbf{x}_{i}$.
Fig.~\ref{fig:user_movement_process} illustrates the \ac{csi} and
user position tracking performance of Algorithm~\ref{alg:radio-map-assisted-switching-Kalman-filter-csi-tracking}.
Specifically, Fig.~\ref{fig:user_movement_process}(a) shows the
channel capacity along a trajectory transitioning through \ac{nlos},
\ac{los}, and \ac{nlos} regions to \ac{bs}~7. The proposed scheme
outperforms other baseline schemes, ensuring smooth tracking during
the \ac{nlos}-\ac{los} transition. Fig.~\ref{fig:user_movement_process}(b)
reveals that the channel capacity efficiency ratio of the proposed
scheme remains above $90\%$ for an \ac{los} trajectory to \ac{bs}~1.
The tracking performance of the conventional \ac{kf} scheme is unstable
due to the lack of prior \ac{csi} statistics, while the \ac{cam}-aided
scheme achieves above $78\%$ of capacity over that of perfect \ac{csi}.
Fig.~\ref{fig:user_movement_process}(c) highlights that localization
errors along the same trajectory in Fig.~\ref{fig:user_movement_process}(a)
stay below $26$~meters, where \ac{bs}~1 operates under \ac{los}
conditions and \ac{bs}~6 under \ac{nlos}. Additionally, the presence
of multiple \acpl{bs} contributes to more stable performance and
reduced localization errors, with values consistently below $14$~meters.
For multi-BS cooperation, each \ac{bs} independently computes the
position distribution using (\ref{eq:bar-pi-pt}), and then, these
distributions are combined multiplicatively across \acpl{bs} to form
a joint posterior. Subsequently, the position ${\bf p}_{t}\in\mathcal{X}$
that maximizes this posterior is selected according to (\ref{eq:tracking-user-position}).

Fig.~\ref{fig:capacity_versus_SNR} shows the channel capacity efficiency
ratio versus \ac{snr} when $7$ \acpl{bs} serve the user. This is
merely an illustrative example to demonstrate the potential performance
gains achievable with multiple \acpl{bs}. The proposed scheme achieves
over $97\%$ of the capacity attained with perfect \ac{csi} at an
\ac{snr} of $20$~dB. Additionally, the proposed scheme nearly matches
the performance of the genius-aided scheme, which assumes perfect
user positions. The one-shot \ac{cam} scheme, which utilizes only
the current measurement and $32$ pilot observations ($M=32$), achieves
approximately $87\%$ of the ideal capacity at an \ac{snr} of $30$~dB.
The \ac{kf} scheme performs well at high \ac{snr} but achieves only
about $36\%$ of the ideal capacity at an \ac{snr} of $10$~dB.
The \ac{ar}-based and \ac{lstm}-based methods perform poorly due
to their dependence on highly accurate historical channel data and
restrictive assumptions about stationary spatial distributions. The
\ac{ls} estimator consistently underperforms other baselines, because
it fails to exploit spatial and temporal priors and relies solely
on a one-shot observation for \ac{csi} estimation.

Fig.~\ref{fig:capacity_versus_BS} depicts the channel capacity efficiency
ratio versus \ac{bs} groups under a $30$~dB \ac{snr}. The randomly
generated user trajectories are \ac{los} to \acpl{bs}~$1$ and
$2$ while they are \ac{nlos} to \acpl{bs}~$4$ and $6$. The proposed
scheme achieves similar performance to the genius-aided scheme. Although
some baseline schemes achieve over $83\%$ of the ideal capacity in
\ac{los} cases, their performance degrades significantly in \ac{nlos}
cases due to the highly complex \ac{csi} distributions caused by
intricate urban topologies.
\begin{figure}
\selectlanguage{american}%
\begin{centering}
\includegraphics[width=0.95\columnwidth]{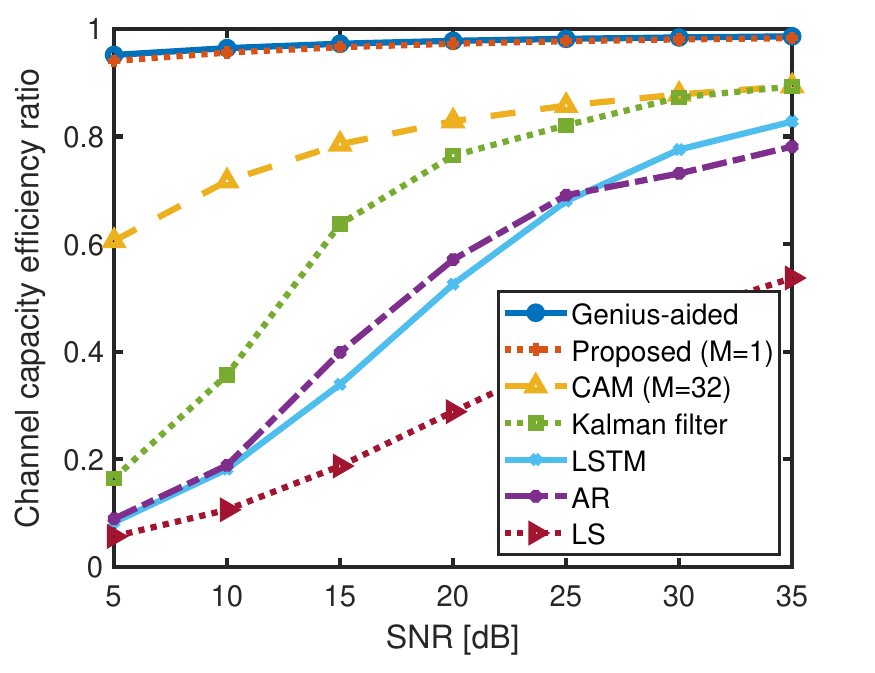}
\par\end{centering}
\caption{\foreignlanguage{english}{Channel capacity efficiency ratio versus \ac{snr}}}

\label{fig:capacity_versus_SNR}\selectlanguage{english}%
\end{figure}

Table~\ref{tab:adaptive-sensing} shows the impact of adaptive sensing
matrices in the proposed CSI tracking scheme. It is observed that
adaptive sensing provides additional performance gains, especially
in low-SNR scenarios. For example, at $10$ dB \ac{snr} under \ac{nlos}
conditions to \ac{bs}~6, the proposed method with adaptive ${\bf A}_{t}$
achieves a channel capacity efficiency ratio of $55\%$, representing
a $14\%$ improvement over the same scheme with random ${\bf A}_{t}$.
In addition, both variants of the proposed scheme with adaptive and
random ${\bf A}_{t}$ achieve over $91\%$ of the ideal capacity under
\ac{los} conditions to \ac{bs} 1, demonstrating the robustness of
the proposed radio-map-embedded CSI tracking scheme, thanks to the
strong spatial prior provided by the radio map.

Table~\ref{tab:Channel-capacity-ratio-gamma-v} indicates that the
coefficient $\gamma$ in the channel model~(\ref{eq:ar_channel_model})
has a negligible impact on the \ac{csi} tracking performance of the
proposed scheme if it is within a certain range, {\em e.g.}, $[0.1,0.7]$.
Moreover, the proposed scheme demonstrates stable and robust performance
across a user speed range of $5$ to $25$ meters per second.

Additionally, Table~\ref{tab:Average-running-time} reports the average
CSI estimation time at BS~1 under the same settings as Fig.~\ref{fig:capacity_versus_BS}.
The proposed scheme takes approximately $0.27$~ms to track the CSI
for BS 1, which is comparable to the CAM-aided scheme that achieves
a similar running time of $0.27$~ms but with a slightly lower channel
capacity efficiency ratio of $83\%$. In contrast, schemes such as
the KF-based and AR-based approaches are slightly faster (about $0.19$~ms
and $0.036$~ms, respectively) but suffer from significant performance
degradation, with over $21\%$ lower capacity efficiency ratios.
\begin{figure}
\selectlanguage{american}%
\begin{centering}
\includegraphics[width=0.92\columnwidth]{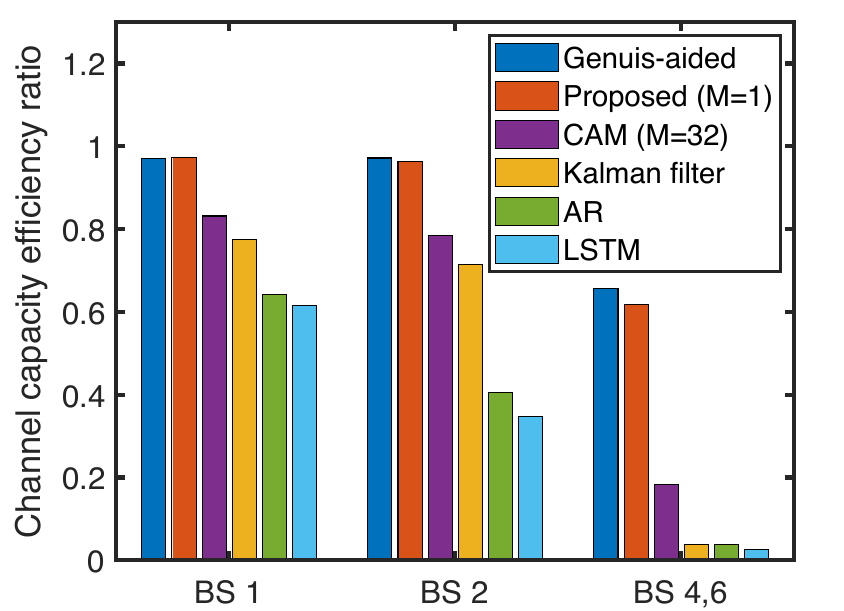}
\par\end{centering}
\caption{\foreignlanguage{english}{Channel capacity efficiency ratio versus \ac{bs} groups}}

\label{fig:capacity_versus_BS}\selectlanguage{english}%
\end{figure}
}
\begin{table}
\selectlanguage{american}%
\caption{\label{tab:adaptive-sensing}\foreignlanguage{english}{Channel capacity
efficiency ratio versus \ac{snr}, with and without adaptive sensing
in the proposed CSI tracking scheme}}

\renewcommand{\arraystretch}{1.4}
\centering{}%
\begin{tabular}{>{\centering}m{0.23\columnwidth}>{\centering}m{0.08\columnwidth}|>{\centering}m{0.07\columnwidth}>{\centering}m{0.07\columnwidth}>{\centering}m{0.07\columnwidth}>{\centering}m{0.07\columnwidth}>{\centering}m{0.07\columnwidth}}
\hline 
\multicolumn{2}{c|}{\foreignlanguage{english}{SNR {[}dB{]}}} & \foreignlanguage{english}{$10$} & \foreignlanguage{english}{$15$} & \foreignlanguage{english}{$20$} & \foreignlanguage{english}{$25$} & \foreignlanguage{english}{$30$}\tabularnewline
\hline 
\multirow{2}{0.23\columnwidth}{\foreignlanguage{english}{\centering{}Adaptive sensing}} & \foreignlanguage{english}{BS 1} & \foreignlanguage{english}{$92\%$} & \foreignlanguage{english}{$93\%$} & \foreignlanguage{english}{$95\%$} & \centering{}$97\%$ & \foreignlanguage{english}{$97\%$}\tabularnewline
\cline{2-7}
 & \foreignlanguage{english}{BS 6} & \foreignlanguage{english}{$55\%$} & \centering{}$55\%$ & \foreignlanguage{english}{$58\%$} & \centering{}$58\%$ & \foreignlanguage{english}{$59\%$}\tabularnewline
\hline 
\multirow{2}{0.23\columnwidth}{\foreignlanguage{english}{\centering{}Random sensing}} & \foreignlanguage{english}{BS 1} & \foreignlanguage{english}{$91\%$} & \foreignlanguage{english}{$93\%$} & \foreignlanguage{english}{$95\%$} & \centering{}$96\%$ & \foreignlanguage{english}{$97\%$}\tabularnewline
\cline{2-7}
 & \foreignlanguage{english}{BS 6} & \foreignlanguage{english}{$41\%$} & \foreignlanguage{english}{$48\%$} & \foreignlanguage{english}{$50\%$} & \foreignlanguage{english}{$50\%$} & \foreignlanguage{english}{$53\%$}\tabularnewline
\hline 
\end{tabular}\selectlanguage{english}%
\end{table}

\selectlanguage{american}%
\begin{table}
\caption{\label{tab:Channel-capacity-ratio-gamma-v}Track performance versus
$\gamma$ $(\bar{v}=10)$ and user speed $\bar{v}$ $(\gamma=0.1)$}

\renewcommand{\arraystretch}{1.4}
\centering{}%
\begin{tabular}{>{\centering}m{0.29\columnwidth}|>{\centering}m{0.12\columnwidth}>{\centering}m{0.12\columnwidth}>{\centering}m{0.12\columnwidth}>{\centering}m{0.12\columnwidth}}
\hline 
\foreignlanguage{english}{\centering{}} & \centering{}\textbf{$\gamma=0.1$} & \foreignlanguage{english}{$\gamma=0.7$} & \foreignlanguage{english}{$\bar{v}=5$} & \foreignlanguage{english}{$\bar{v}=25$}\tabularnewline
\hline 
Channel capacity efficiency ratio & \centering{}$98.0\%$ & \foreignlanguage{english}{$97.9\%$} & \foreignlanguage{english}{$98.7\%$} & \foreignlanguage{english}{$97.7\%$}\tabularnewline
\hline 
\end{tabular}
\end{table}
\begin{table}
\caption{\foreignlanguage{english}{\label{tab:Average-running-time}Average running time for CSI estimation
{[}millisecond{]}}}

\renewcommand{\arraystretch}{1.2}
\centering{}%
\begin{tabular}{>{\centering}m{0.12\columnwidth}>{\centering}m{0.12\columnwidth}>{\centering}m{0.12\columnwidth}>{\centering}m{0.1\columnwidth}>{\centering}m{0.1\columnwidth}>{\centering}m{0.1\columnwidth}}
\toprule 
\foreignlanguage{english}{\centering{}\textbf{Genius}} & \centering{}\textbf{Proposed} & \foreignlanguage{english}{\textbf{CAM}} & \foreignlanguage{english}{\textbf{KF}} & \foreignlanguage{english}{\textbf{AR}} & \foreignlanguage{english}{\textbf{LSTM}}\tabularnewline
\cmidrule(lr){1-1}\cmidrule{2-2}\cmidrule(lr){3-3}\cmidrule(lr){4-4}\cmidrule(lr){5-5}\cmidrule{6-6}
\foreignlanguage{english}{$0.21$} & \foreignlanguage{english}{$0.27$} & \centering{}$0.27$ & \foreignlanguage{english}{$0.19$} & \foreignlanguage{english}{$0.036$} & \foreignlanguage{english}{$1.6$}\tabularnewline
\bottomrule
\end{tabular}
\end{table}

\selectlanguage{english}%

\subsection{Radio Map Construction}

\label{subsec:Exp_Blind-Radio-Map-Construction-Trajectory-Discovery}
\begin{figure}
\begin{centering}
\subfigure[Estimated user trajectory]{\includegraphics[width=0.49\columnwidth]{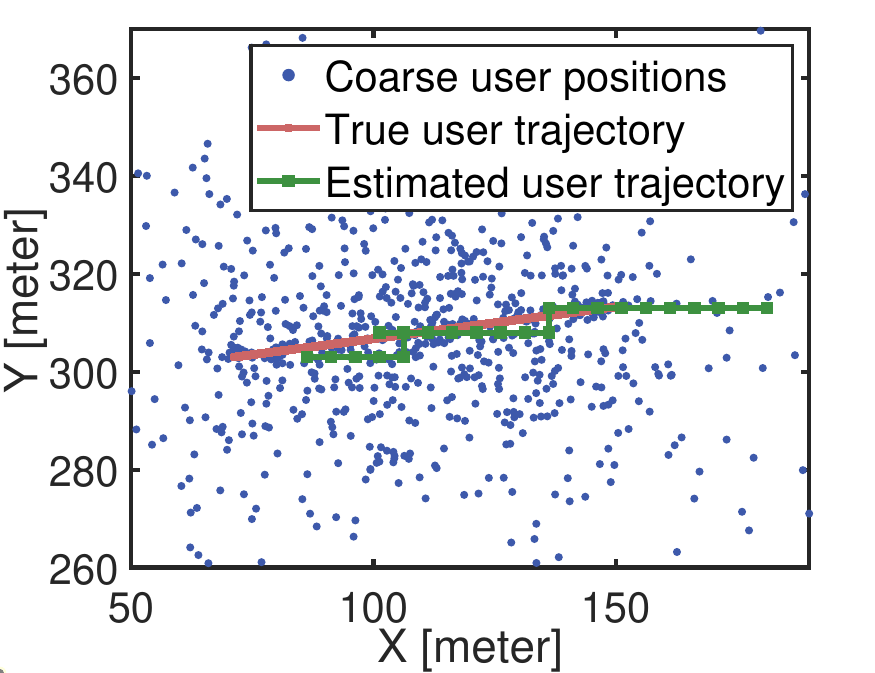}}\subfigure[Mean localization error]{\includegraphics[width=0.5\columnwidth]{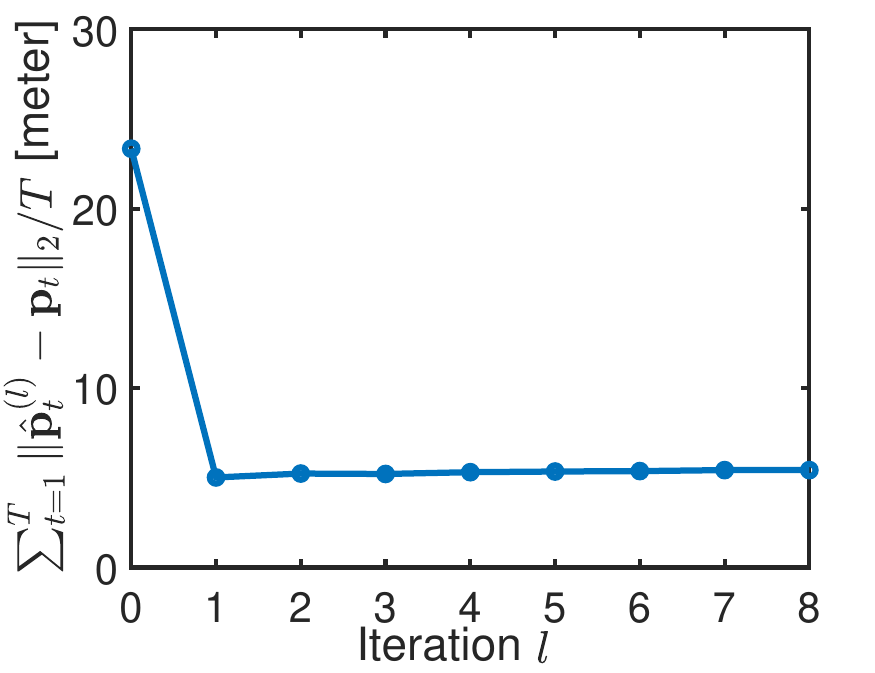}}
\par\end{centering}
\caption{\foreignlanguage{american}{User localization and convergence of Algorithm~\ref{alg:radio-map-construction}.}}

\label{fig:convergence-alg2}
\end{figure}

In this subsection, we construct radio maps for each \ac{bs} from
sparse and unlabeled pilot observation sequences in an offline manner.
This is a one-time process performed prior to deployment. Fig.~\ref{fig:convergence-alg2}(a)
shows an example of the estimated user trajectory obtained by solving
Problem~(\ref{eq:localization-problem}). The estimated trajectory
matches the ground truth well, even though the coarse user locations
are scattered over an area hundreds of meters wide. Numerical results
indicate a localization error of less than $6$~meters when \ac{bs}~$1$
serves the user. Fig.~\ref{fig:convergence-alg2}(b) confirms the
convergence of Algorithm~\ref{alg:radio-map-construction} with localization
errors staying stable as the number of iterations increases.
\begin{table}
\selectlanguage{american}%
\caption{\label{tab:NMSE-C}\foreignlanguage{english}{Normalized $L_{2}$-norm
estimation error of channel covariance versus number of measurements
$T_{\mathrm{total}}$ and dimension $M$}}

\renewcommand{\arraystretch}{1.4}
\centering{}%
\begin{tabular}{>{\centering}m{0.3\columnwidth}|>{\centering}m{0.08\columnwidth}>{\centering}m{0.08\columnwidth}>{\centering}m{0.08\columnwidth}>{\centering}m{0.08\columnwidth}>{\centering}m{0.08\columnwidth}}
\hline 
\foreignlanguage{english}{\centering{}$T_{\mathrm{total}}$} & \centering{}\textbf{$35000$} & \foreignlanguage{english}{\centering{}$70000$} & \foreignlanguage{english}{$105000$} & \foreignlanguage{english}{$140000$} & \foreignlanguage{english}{$175000$}\tabularnewline
\hline 
\foreignlanguage{english}{$M=1$, BS 1 (LOS)} & \centering{}$0.837$ & \centering{}$0.825$ & \foreignlanguage{english}{$0.819$} & \foreignlanguage{english}{$0.812$} & \foreignlanguage{english}{$0.806$}\tabularnewline
\cline{1-1}
\foreignlanguage{english}{$M=4$, BS 1 (LOS)} & \foreignlanguage{english}{$0.766$} & \foreignlanguage{english}{$0.729$} & \foreignlanguage{english}{$0.706$} & \foreignlanguage{english}{$0.691$} & \foreignlanguage{english}{$0.676$}\tabularnewline
\cline{1-1}
\multicolumn{1}{c|}{\foreignlanguage{english}{$M=16$, BS 1 (LOS)}} & \foreignlanguage{english}{$0.630$} & \foreignlanguage{english}{$0.617$} & \foreignlanguage{english}{$0.614$} & \foreignlanguage{english}{$0.615$} & \foreignlanguage{english}{$0.612$}\tabularnewline
\cline{1-1}
\multicolumn{1}{c|}{\foreignlanguage{english}{$M=16$, BS 6 (NLOS)}} & \foreignlanguage{english}{$0.781$} & \foreignlanguage{english}{$0.762$} & \foreignlanguage{english}{$0.755$} & \foreignlanguage{english}{$0.748$} & \foreignlanguage{english}{$0.742$}\tabularnewline
\hline 
\end{tabular}\selectlanguage{english}%
\end{table}

\selectlanguage{american}%
\begin{table}
\caption{\label{tab:Normalized-projection-error}\foreignlanguage{english}{Normalized
projection error of channel covariance versus number of measurements
$T_{\mathrm{total}}$ and dimension $M$}}

\renewcommand{\arraystretch}{1.4}
\centering{}%
\begin{tabular}{>{\centering}m{0.3\columnwidth}|>{\centering}m{0.08\columnwidth}>{\centering}m{0.08\columnwidth}>{\centering}m{0.08\columnwidth}>{\centering}m{0.08\columnwidth}>{\centering}m{0.08\columnwidth}}
\hline 
\foreignlanguage{english}{\centering{}$T_{\mathrm{total}}$} & \centering{}\textbf{$35000$} & \foreignlanguage{english}{\centering{}$70000$} & \foreignlanguage{english}{$105000$} & \foreignlanguage{english}{$210000$} & \foreignlanguage{english}{$280000$}\tabularnewline
\hline 
\foreignlanguage{english}{$M=1$, BS 1 (LOS)} & \centering{}$0.227$ & \centering{}$0.227$ & \foreignlanguage{english}{$0.203$} & \foreignlanguage{english}{$0.139$} & \foreignlanguage{english}{$0.131$}\tabularnewline
\cline{1-1}
\foreignlanguage{english}{$M=4$, BS 1 (LOS)} & \foreignlanguage{english}{$0.235$} & \foreignlanguage{english}{$0.207$} & \foreignlanguage{english}{$0.199$} & \foreignlanguage{english}{$0.131$} & \foreignlanguage{english}{$0.131$}\tabularnewline
\cline{1-1}
\multicolumn{1}{c|}{\foreignlanguage{english}{$M=16$, BS 1 (LOS)}} & \foreignlanguage{english}{$0.211$} & \foreignlanguage{english}{$0.131$} & \foreignlanguage{english}{$0.092$} & \foreignlanguage{english}{$0.064$} & \foreignlanguage{english}{$0.052$}\tabularnewline
\cline{1-1}
\multicolumn{1}{c|}{\foreignlanguage{english}{$M=16$, BS 6 (NLOS)}} & \foreignlanguage{english}{$0.214$} & \foreignlanguage{english}{$0.202$} & \foreignlanguage{english}{$0.202$} & \foreignlanguage{english}{$0.198$} & \foreignlanguage{english}{$0.190$}\tabularnewline
\hline 
\end{tabular}
\end{table}
\begin{figure}[t]
\selectlanguage{english}%
\begin{centering}
\includegraphics[width=0.95\columnwidth]{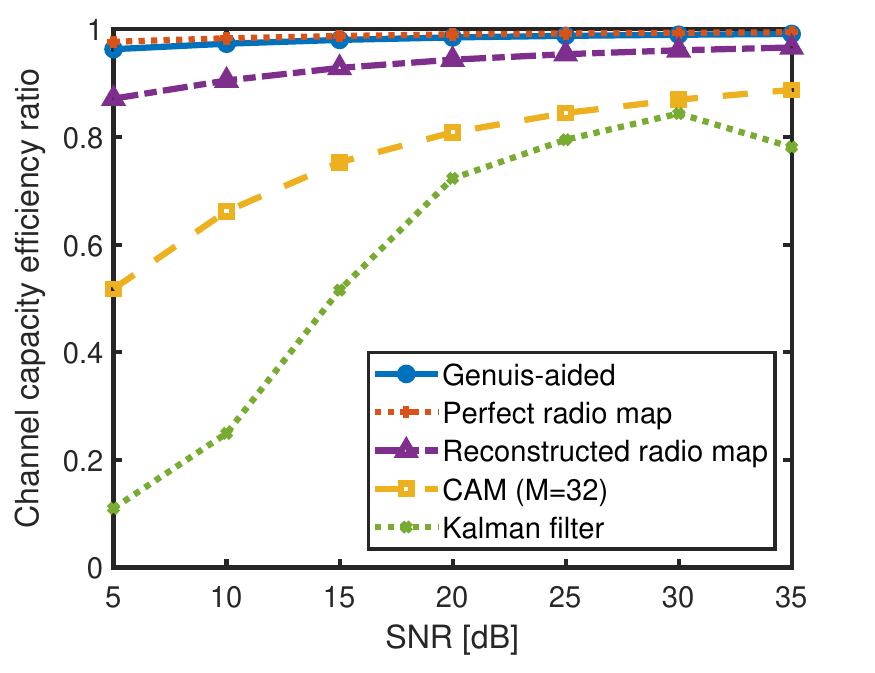}
\par\end{centering}
\caption{\label{fig:capacity_RM_construction}Channel capacity efficiency ratio
for \ac{los} trajectories to \ac{bs}~1}
\selectlanguage{american}%
\end{figure}

\selectlanguage{english}%
\begin{figure}[t]
\begin{centering}
\includegraphics[width=0.95\columnwidth]{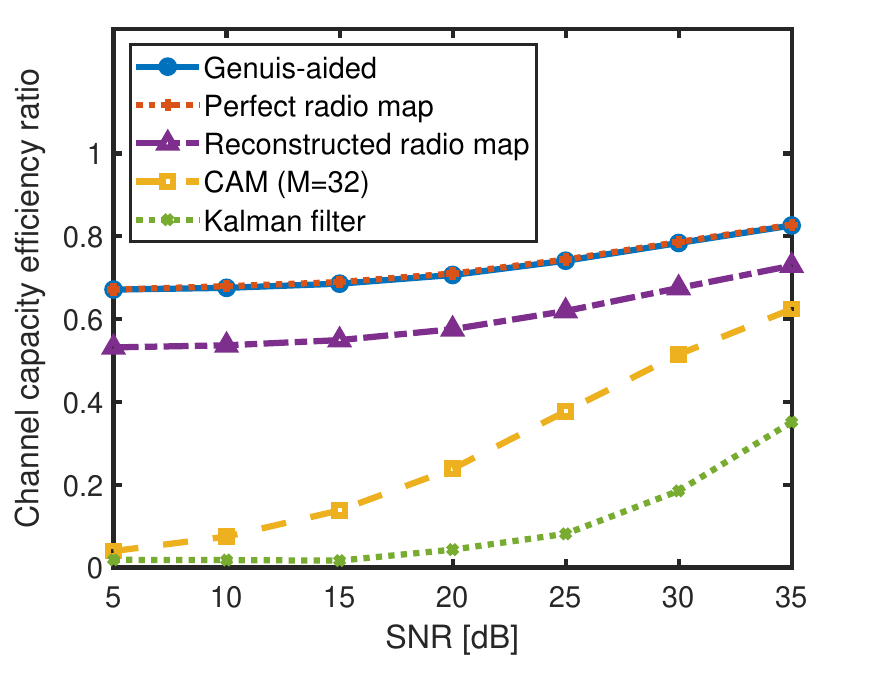}
\par\end{centering}
\caption{\label{fig:capacity_RM_construction-nlos}Channel capacity efficiency
ratio for \ac{nlos} trajectories to \ac{bs}~6}
\end{figure}
Table~\foreignlanguage{american}{\ref{tab:NMSE-C}} presents the
normalized $L_{2}$-norm estimation error of the channel covariance
versus the number of measurements $T_{\mathrm{total}}$ under different
settings of $M$. Here, multiple sequences of channel measurements
are collected, where each sequence consists of 350 measurements, and
$T_{{\rm total}}$ denotes the total number of measurements across
all sequences. Larger values of $M$ and $T_{\mathrm{total}}$ reduce
the normalized $L_{2}$-norm estimation error as supported by Proposition~\ref{prop:upper-bound-of-the-estimation-error}.
It is observed that when $M$ is significantly smaller than $N_{\text{t}}$,
{\em e.g.}, $M=4$, it requires $175,000$ measurements to reach
an estimation error of $0.676$, while it only requires less than
$35,000$ measurements to achieve a smaller estimation error when
$M=16$. Additionally, constructing radio maps for \acpl{bs} under
\ac{nlos} conditions is more challenging compared to \acpl{bs} under
\ac{los} due to the rich multipath effects.

Table~\ref{tab:Normalized-projection-error} further illustrates
the normalized projection error of channel covariance versus the number
of measurements $T_{\mathrm{total}}$ and dimension $M$. The results
demonstrate that increased measurements or dimension of pilot observations
consistently reduce the normalized projection error for both BS~1
and BS~6. However, BS~1, operating under LOS conditions, achieves
a significantly lower normalized projection error with sufficient
measurements. Specifically, BS~1 reaches a normalized projection
error of $0.211$ with $35,000$ measurements and $0.052$ with $280,000$
measurements. In contrast, BS~6, operating under NLOS conditions,
only achieves a normalized projection error of $0.190$ even with
$280,000$ measurements. These findings indicate that radio map construction
under NLOS conditions presents greater challenges than under LOS conditions
due to rich multipath effects.

Fig.~\ref{fig:capacity_RM_construction} and Fig.~\ref{fig:capacity_RM_construction-nlos}
illustrate the \ac{csi} tracking performance for \ac{bs}~$1$ under
\ac{los} conditions and \ac{bs}~$6$ under \ac{nlos}, based on
the radio maps that are reconstructed using $280,000$ and $420,000$
\ac{csi} measurements, respectively with $M=16$. Using reconstructed
radio maps, the proposed \ac{csi} tracking scheme with one pilot
observation at each time slot and a total number of $100$ pilot observations
achieves above $88\%$ of the ideal capacity for \ac{bs}~1, and
the proposed scheme outperforms other baseline schemes in both low-\ac{snr}
and high-\ac{snr} scenarios for \ac{bs}~$6$.

Additionally, the memory usage of the proposed radio map is assessed.
The radio map covers an area of $740$~m$\times$$710$~m with a
grid resolution of $5$~m, resulting in $21,016$ grids. Each grid
stores its center coordinates as double-precision values and a $64\times64$
complex covariance matrix in single precision. The theoretical storage
requirement is approximately $657$~MB, and the actual uncompressed
size is about $659$~MB. This demonstrates the feasibility of the
proposed method, as it achieves a reasonable memory size while maintaining
the necessary precision for the radio map.

\section{Conclusion}

\label{sec:Conclusion}

This paper addresses the challenges of \ac{csi} tracking in massive
\ac{mimo} systems, particularly in the absence of stationary \ac{csi}
statistics and precise location information. By integrating radio
maps with an adaptive \ac{skf} framework, we track the \ac{csi}
and the location using an extremely sparse pilot observation sequence
without location labels. For radio map construction, the joint estimation
of location sequences and channel covariance matrices is facilitated
by an \ac{hmm}, and an unbiased channel covariance estimator is found
with the estimation error analyzed. Numerical results show that the
proposed radio-map-embedded \ac{csi} tracking scheme delivers over
$97\%$ of the capacity of perfect \ac{csi} with reduced pilot observations
at a $20$~dB \ac{snr} under \ac{los}, while a conventional \ac{kf}
can only achieve $76\%$. In addition, the proposed algorithm reduces
localization errors from $30$ meters from the prior to $6$ meters
for radio map construction in our framework.


\appendices{}


\section{Proof of Proposition~\ref{prop:Factorization-of-logp}}

\label{sec:Proof-of-Proposition-Factorization}

According to Bayes' theorem, $p(\mathcal{Y}_{T},\mathcal{H}_{T})$
is derived as
\begin{align}
p(\mathcal{Y}_{T},\mathcal{H}_{T}) & =p(\mathbf{y}_{T}|\mathcal{Y}_{T-1},\mathcal{H}_{T})p(\mathcal{Y}_{T-1},\mathcal{H}_{T})\nonumber \\
 & =p(\mathbf{y}_{T}|\mathbf{h}_{T})P(\mathbf{h}_{T}|\mathcal{Y}_{T-1},\mathcal{H}_{T-1})p(\mathcal{Y}_{T-1},\mathcal{H}_{T-1})\label{eq:joint_prop_y_h-step1}\\
 & =p(\mathbf{y}_{T}|\mathbf{h}_{T})P(\mathbf{h}_{T}|\mathbf{h}_{T-1},\mathbf{p}_{T})p(\mathcal{Y}_{T-1},\mathcal{H}_{T-1})\label{eq:joint_prop_y_h-step4}\\
 & =\prod_{t=1}^{T}p(\mathbf{y}_{\text{\ensuremath{t}}}|\mathbf{h}_{t})\times\prod_{t=2}^{T}p(\mathbf{h}_{\text{\ensuremath{t}}}|\mathbf{h}_{t-1},{\bf p}_{t})p(\mathbf{h}_{1})\label{eq:joint_prop_y_h}
\end{align}
where equations~(\ref{eq:joint_prop_y_h-step1}) and (\ref{eq:joint_prop_y_h-step4})
hold because $\mathbf{y}_{T}$ depends on $\mathbf{h}_{T}$, and $\mathbf{h}_{T}$
depends on $\mathbf{h}_{T-1}$ and $\mathbf{p}_{T}$, and equation~(\ref{eq:joint_prop_y_h})
is derived by recursively decomposing $p(\mathcal{Y}_{T-1},\mathcal{H}_{T-1})$
in the same manner as $p(\mathcal{Y}_{T},\mathcal{H}_{T})$. Thus,
$p(\mathcal{Y}_{T}|\mathcal{H}_{T})$ is derived as
\begin{align}
p(\mathcal{Y}_{T}|\mathcal{H}_{T}) & =p(\mathcal{Y}_{T},\mathcal{H}_{T})/p(\mathcal{H}_{T})=\prod_{t=1}^{T}p(\mathbf{y}_{\text{\ensuremath{t}}}|\mathbf{h}_{t})\label{eq:prop_Y_given_HT}
\end{align}

Similar to (\ref{eq:joint_prop_y_h}), $p(\mathcal{H}_{T},\mathcal{P}_{T})$
is derived as 
\begin{flalign}
p(\mathcal{H}_{T},\mathcal{P}_{T}) & \text{\ensuremath{=}}\prod_{t=2}^{T}p({\bf h}_{t}|{\bf h}_{t-1},{\bf p}_{t})\mathbb{P}(\mathbf{p}_{t}|\mathbf{p}_{t-1})p({\bf h}_{1})\mathbb{P}(\mathbf{p}_{1}).\label{eq:recursive-derivation}
\end{flalign}

Based on (\ref{eq:prop_Y_given_HT}) and (\ref{eq:recursive-derivation}),
$p(\mathcal{Y}_{T},\mathcal{P}_{T},\mathcal{H}_{T})$ is factorized
as
\begin{align}
p(\mathcal{Y}_{T},\mathcal{P}_{T},\mathcal{H}_{T}) & =p(\mathcal{Y}_{T}|\mathcal{P}_{T},\mathcal{H}_{T})p(\mathcal{H}_{T},\mathcal{P}_{T}).\nonumber \\
 & =p(\mathcal{Y}_{T}|\mathcal{H}_{T})p(\mathcal{H}_{T},\mathcal{P}_{T})\label{eq:joint-probability-Y-P-H-step2}\\
 & =\prod_{t=1}^{T}p({\bf y}_{t}|{\bf h}_{t})\prod_{t=2}^{T}p({\bf h}_{t}|{\bf h}_{t-1},{\bf p}_{t})p({\bf h}_{1})\nonumber \\
 & \quad\times\prod_{t=2}^{T}\mathbb{P}(\mathbf{p}_{t}|\mathbf{p}_{t-1})\mathbb{P}(\mathbf{p}_{1})\label{eq:joint-probability-Y-P-H}
\end{align}
where (\ref{eq:joint-probability-Y-P-H-step2}) holds because ${\bf y}_{t}$
is independent of $\mathbf{p}_{t}$ given ${\bf h}_{t}$.

\section{Proof of Proposition~\ref{prop:error-propagation-equation}}

\label{sec:Proof-of-Proposition-Error-Propagation-Equation}

The error propagation process can be proved by validating ${\bf Q}_{t}{\bf Q}_{t}^{-1}=\mathbf{I},$
where ${\bf Q}_{t}$ is in (\ref{eq:update-error-covariance}) and
${\bf Q}_{t}^{-1}$ is in (\ref{eq:alternative-error-covariance-update}).

First, substituting the optimal $\mathbf{K}_{t}$ in (\ref{eq:Kalman-gain})
into (\ref{eq:update-error-covariance}), we have
\begin{equation}
{\bf Q}_{t}=\mathbf{Q}_{t|t-1}-\mathbf{Q}_{t|t-1}{\bf A}_{t}^{\text{H}}({\bf A}_{t}\mathbf{Q}_{t|t-1}{\bf A}_{t}^{\text{H}}+\sigma_{\text{n}}^{2}{\bf I})^{-1}{\bf A}_{t}\mathbf{Q}_{t|t-1}.\label{eq:another-form-Q-t}
\end{equation}

Second, ${\bf Q}_{t}{\bf Q}_{t}^{-1}$ is derived as
\begin{align*}
{\bf Q}_{t}{\bf Q}_{t}^{-1}= & (\mathbf{Q}_{t|t-1}-\mathbf{Q}_{t|t-1}{\bf A}_{t}^{\text{H}}({\bf A}_{t}\mathbf{Q}_{t|t-1}{\bf A}_{t}^{\text{H}}+\sigma_{\text{n}}^{2}{\bf I})^{-1}\\
 & \times{\bf A}_{t}\mathbf{Q}_{t|t-1})(\mathbf{Q}_{t|t-1}^{-1}+\mathbf{A}_{t}^{\text{H}}(\sigma_{\mathrm{n}}^{2}\mathbf{I})^{-1}\mathbf{A}_{t})\\
= & \mathbf{I}-\mathbf{Q}_{t|t-1}{\bf A}_{t}^{\text{H}}[({\bf A}_{t}\mathbf{Q}_{t|t-1}{\bf A}_{t}^{\text{H}}+\sigma_{\text{n}}^{2}{\bf I})^{-1}-(\sigma_{\mathrm{n}}^{2}\mathbf{I})^{-1}\\
 & +({\bf A}_{t}\mathbf{Q}_{t|t-1}{\bf A}_{t}^{\text{H}}+\sigma_{\text{n}}^{2}{\bf I})^{-1}{\bf A}_{t}\mathbf{Q}_{t|t-1}\mathbf{A}_{t}^{\text{H}}(\sigma_{\mathrm{n}}^{2}\mathbf{I})^{-1}]{\bf A}_{t}\\
= & \mathbf{I}-\mathbf{Q}_{t|t-1}{\bf A}_{t}^{\text{H}}[(\mathbf{\sigma_{\text{n}}^{2}{\bf I}})^{-1}-(\sigma_{\mathrm{n}}^{2}\mathbf{I})^{-1}]\mathbf{A}_{t}\\
= & \mathbf{I}.
\end{align*}

Thus, $\mathbf{Q}_{t}^{-1}$ in (\ref{eq:alternative-error-covariance-update})
is the inverse of $\mathbf{Q}_{t}$.

\section{Proof of Proposition~\ref{prop:The-solution-to-At}}

\label{sec:Proof-of-Proposition-Solution-At}

Using \ac{evd} of $\mathbf{Q}_{t|t-1}$, we have
\begin{equation}
\mathbf{Q}_{t|t-1}=\mathbf{W}_{t-1}\boldsymbol{\Lambda}_{t-1}\mathbf{W}_{t-1}^{\mathrm{H}}\label{eq:svd-Q-t-1}
\end{equation}
where $\boldsymbol{\Lambda}_{t-1}=\text{diag}\{\lambda_{t-1,1},\dots,\lambda_{t-1,N_{\mathrm{t}}}\}$
is a diagonal matrix containing the eigenvalues of $\mathbf{Q}_{t|t-1}$
in descending order and $\mathbf{W}_{t-1}=[\mathbf{w}_{t-1,1},\mathbf{w}_{t-1,2},\dots,\mathbf{w}_{t-1,N_{\mathrm{t}}}]$
is a unitary matrix with orthogonal eigenvectors of $\mathbf{Q}_{t|t-1}$.

Define $\mathbf{B}_{t}\triangleq\mathbf{A}_{t}\mathbf{W}_{t-1}\in\mathbb{C}^{M\times N_{\mathrm{t}}}$.
Based on (\ref{eq:svd-Q-t-1}), $|\mathbf{I}+\sigma_{\text{n}}^{-2}\mathbf{A}_{t}\mathbf{Q}_{t|t-1}\mathbf{A}_{t}^{\text{H}}|$
is derived as
\begin{align}
|\mathbf{I}+\sigma_{\text{n}}^{-2}\mathbf{A}_{t}\mathbf{Q}_{t|t-1}\mathbf{A}_{t}^{\text{H}}| & =|\mathbf{I}+\sigma_{\text{n}}^{-2}\mathbf{A}_{t}\mathbf{W}_{t-1}\boldsymbol{\Lambda}_{t-1}\mathbf{W}_{t-1}^{\mathrm{H}}\mathbf{A}_{t}^{\text{H}}|\nonumber \\
 & =|\mathbf{I}+\sigma_{\text{n}}^{-2}\mathbf{B}_{t}\boldsymbol{\Lambda}_{t-1}\mathbf{B}_{t}^{\mathrm{H}}|\nonumber \\
 & =|\mathbf{B}_{t}(\mathbf{I}+\sigma_{\text{n}}^{-2}\boldsymbol{\Lambda}_{t-1})\mathbf{B}_{t}^{\mathrm{H}}|.
\end{align}

The matrix $\mathbf{B}_{t}(\mathbf{I}+\sigma_{\text{n}}^{-2}\boldsymbol{\Lambda}_{t-1})\mathbf{B}_{t}^{\mathrm{H}}$
is a projection of the diagonal matrix $\mathbf{I}+\sigma_{\text{n}}^{-2}\boldsymbol{\Lambda}_{t-1}$
onto the subspace defined by $\mathbf{B}_{t}$ with $\mathbf{B}_{t}\mathbf{B}_{t}^{\text{H}}=\mathbf{A}_{t}\mathbf{W}_{t-1}(\mathbf{A}_{t}\mathbf{W}_{t-1})^{\text{H}}=\mathbf{I}$.
To maximize $|\mathbf{B}_{t}(\mathbf{I}+\sigma_{\text{n}}^{-2}\boldsymbol{\Lambda}_{t-1})\mathbf{B}_{t}^{\mathrm{H}}|$,
$\mathbf{B}_{t}$ should select the top $M$ eigenvalues of $\mathbf{I}+\sigma_{\text{n}}^{-2}\boldsymbol{\Lambda}_{t-1}$.

Thus, $|\mathbf{B}_{t}(\mathbf{I}+\sigma_{\text{n}}^{-2}\boldsymbol{\Lambda}_{t-1})\mathbf{B}_{t}^{\mathrm{H}}|$
reaches its maximum as $\prod_{i=1}^{M}(1+\sigma_{\text{n}}^{2}\lambda_{t-1,i}),$
where $\lambda_{t-1,i}$ for $i=1,2,\dots,M$ are the top $M$ eigenvalues
of $\mathbf{Q}_{t|t-1}$. In this case, $\mathbf{A}_{t}$ is given
by $\mathbf{A}_{t}=[\mathbf{w}_{t-1,1},\mathbf{w}_{t-1,2},\dots,\mathbf{w}_{t-1,M}]^{\text{H}}$
such that $\mathbf{B}_{t}=\mathbf{A}_{t}\mathbf{W}_{t-1}$ is a diagonal
matrix with its first $M$ diagonal elements equal to $1$ and other
elements equal to $0$.
\selectlanguage{american}%

\section{Proof of Proposition~\ref{prop:estimator-of-sample-covariance}}

\label{sec:Proof-of-Proposition-estimator-C}

\selectlanguage{english}%
It suffices to consider $\mathcal{T}_{i}=\{1\}$ and $t=1$ as the
result will follow by linearity of expectation in each discretize
region. Let $\mathbf{h}=\mathbf{h}_{1}$, $\boldsymbol{\Phi}=\boldsymbol{\Phi}_{1}$,
$\mathbf{n}=\mathbf{n}_{1}$, and $\mathbf{A}=\mathbf{A}_{1}$.

First, $\mathbb{E}\{\hat{\boldsymbol{\Omega}}_{\text{y}}(\mathbf{x}_{i})|\mathbf{h}\}$
is expressed as
\begin{align*}
\mathbb{E}\{\hat{\boldsymbol{\Omega}}_{\text{y}}(\mathbf{x}_{i})|\mathbf{h}\} & =\frac{N_{\mathrm{t}}^{2}}{M^{2}}\mathbb{E}\{(\boldsymbol{\Phi}\mathbf{h}+\mathbf{A}^{\text{H}}\mathbf{n})(\boldsymbol{\Phi}\mathbf{h}+\mathbf{A}^{\text{H}}\mathbf{n})^{\text{H}}|\mathbf{h}\}\\
 & =\frac{N_{\mathrm{t}}^{2}}{M^{2}}(\mathbb{E}\{\boldsymbol{\Phi}\mathbf{h}(\boldsymbol{\Phi}\mathbf{h})^{\text{H}}|\mathbf{h}\}+\mathbb{E}\{\mathbf{A}^{\text{H}}\mathbf{n}\mathbf{n}^{\mathrm{H}}\mathbf{A}\})
\end{align*}
where the cross terms are dropped because of the independence of the
signal and the noise.

Second, the following fact is adopted from \cite[Fact 1]{AziKriSin:J18}.
\begin{fact}
If $\mathbf{x}\in\mathbb{C}^{N_{\mathrm{t}}}$ and $\boldsymbol{\Phi}\in\mathbb{C}^{N_{\mathrm{t}}\times N_{\mathrm{t}}}$
is a uniformly distributed rank $m$ orthogonal projection, then
\[
\boldsymbol{\Phi}\mathbf{x}\overset{\mathrm{d}}{=}\omega\mathbf{x}+\sqrt{\omega-\omega^{2}}\|\mathbf{x}\|\mathbf{W}\alpha
\]
where $\omega\sim\text{Beta}(\frac{M}{2},\frac{N_{\mathrm{t}}-M}{2}),$
$\alpha\in\mathbb{C}^{N_{\mathrm{t}}-1}$ is distributed uniformly
on the unit sphere and independent from $\omega$, and $\mathbf{W}\in\mathbb{C}^{N_{\mathrm{t}}\times(N_{\mathrm{t}}-1)}$
is an orthonormal basis for the subspace orthogonal to $\mathbf{x}$,
{\em i.e.}, $\mathbf{x}^{\text{H}}\mathbf{W}=0$.
\end{fact}
Based on Fact~1 and the expectation of Beta distribution, $\mathbb{E}\{\boldsymbol{\Phi}\mathbf{h}(\boldsymbol{\Phi}\mathbf{h})^{\text{H}}|\mathbf{h}\}$
can be derived and simplified as
\begin{align}
\mathbb{E}\{\boldsymbol{\Phi}\mathbf{h}(\boldsymbol{\Phi}\mathbf{h})^{\text{H}}|\mathbf{h}\}= & \frac{M(MN_{\mathrm{t}}+N_{\mathrm{t}}-2)\mathbf{h}\mathbf{h}^{\mathrm{H}}}{N_{\mathrm{t}}(N_{\mathrm{t}}+2)(N_{\mathrm{t}}-1)}\nonumber \\
 & +\frac{M(N_{\mathrm{t}}-M)\text{tr}(\mathbf{h}\mathbf{h}^{\mathrm{H}})\mathbf{I}}{N_{\mathrm{t}}(N_{\mathrm{t}}+2)(N_{\mathrm{t}}-1)}.\label{eq:exp-Phi-h-Phi-h}
\end{align}

Third, since $\mathbf{n}\sim\mathcal{CN}(0,\sigma_{\text{n}}^{2}\mathbf{I})$
and $\mathbf{n}$ is independent of $\mathbf{A}$, the expectation
$\mathbb{E}\{\mathbf{A}^{\text{H}}\mathbf{n}\mathbf{n}^{\mathrm{H}}\mathbf{A}\}$
is derived as
\begin{align}
\mathbb{E}\{\mathbf{A}^{\text{H}}\mathbf{n}\mathbf{n}^{\mathrm{H}}\mathbf{A}\}= & \mathbf{A}^{\mathrm{H}}\mathbb{E}\{\mathbf{n}\mathbf{n}^{\mathrm{H}}\}\mathbf{A}=\sigma_{\text{n}}^{2}\mathbf{A}^{\mathrm{H}}\mathbf{A}.\label{eq:Phi-A-n-A-n}
\end{align}

Fourth, by taking the expectation of $\mathbb{E}\{\hat{\boldsymbol{\Omega}}_{\text{y}}(\mathbf{x}_{i})|\mathbf{h}\}$
over $\mathbf{h}$ and rescaling (\ref{eq:exp-Phi-h-Phi-h}) and (\ref{eq:Phi-A-n-A-n})
by $N_{\mathrm{t}}^{2}/M^{2}$, $\mathbb{E}\{\hat{\boldsymbol{\Omega}}_{\text{y}}(\mathbf{x}_{i})\}$
is derived as
\begin{align}
\mathbb{E}\{\hat{\boldsymbol{\Omega}}_{\text{y}}(\mathbf{x}_{i})\}= & \frac{N_{\mathrm{t}}(MN_{\mathrm{t}}+N_{\mathrm{t}}-2)}{M(N_{\mathrm{t}}+2)(N_{\mathrm{t}}-1)}\mathbf{C}(\mathbf{x}_{i})\nonumber \\
 & +\frac{N_{\mathrm{t}}(N_{\mathrm{t}}-M)\text{tr}(\mathbf{C}(\mathbf{x}_{i}))\mathbf{I}}{M(N_{\mathrm{t}}+2)(N_{\mathrm{t}}-1)}\nonumber \\
 & +\frac{N_{\mathrm{t}}^{2}\sigma_{\text{n}}^{2}}{M^{2}|\mathcal{T}_{i}|}\sum_{t\in\mathcal{T}_{i}}\mathbf{A}_{t}^{\mathrm{H}}\mathbf{A}_{t}.\label{eq:exptation-hat-Omega-y}
\end{align}

Fifth, by taking the trace operation to both sides of (\ref{eq:exptation-hat-Omega-y}),
and extracting the term $\text{tr}(\mathbf{C}(\mathbf{x}_{i}))$,
we have
\begin{align*}
\text{tr}(\mathbf{C}(\mathbf{x}_{i})) & =\frac{M}{N_{\mathrm{t}}}\text{tr}(\mathbb{E}\{\hat{\boldsymbol{\Omega}}_{\text{y}}(\mathbf{x}_{i})\})-\frac{N_{\mathrm{t}}\sigma_{\text{n}}^{2}}{M|\mathcal{T}_{i}|}\sum_{t\in\mathcal{T}_{i}}\text{tr}(\mathbf{A}_{t}^{\mathrm{H}}\mathbf{A}_{t}).
\end{align*}

Substituting the above $\text{tr}(\mathbf{C}(\mathbf{x}_{i}))$ into
(\ref{eq:exptation-hat-Omega-y}), we can obtain
\begin{align}
\mathbb{E}\{\hat{\boldsymbol{\Omega}}_{\text{y}}(\mathbf{x}_{i})\}= & \frac{N_{\mathrm{t}}(MN_{\mathrm{t}}+N_{\mathrm{t}}-2)}{M(N_{\mathrm{t}}+2)(N_{\mathrm{t}}-1)}\mathbf{C}(\mathbf{x}_{i})\nonumber \\
 & +\frac{(N_{\mathrm{t}}-M)}{(N_{\mathrm{t}}+2)(N_{\mathrm{t}}-1)}\text{tr}(\mathbb{E}\{\hat{\boldsymbol{\Omega}}_{\text{y}}(\mathbf{x}_{i})\})\mathbf{I}\nonumber \\
 & -\frac{N_{\mathrm{t}}^{2}(N_{\mathrm{t}}-M)\sigma_{\text{n}}^{2}}{M^{2}(N_{\mathrm{t}}+2)(N_{\mathrm{t}}-1)|\mathcal{T}_{i}|}\sum_{t\in\mathcal{T}_{i}}\text{tr}(\mathbf{A}_{t}^{\mathrm{H}}\mathbf{A}_{t})\mathbf{I}\nonumber \\
 & +\frac{N_{\mathrm{t}}^{2}\sigma_{\text{n}}^{2}}{M^{2}|\mathcal{T}_{i}|}\sum_{t\in\mathcal{T}_{i}}\mathbf{A}_{t}^{\mathrm{H}}\mathbf{A}_{t}.\label{eq:expectation-Omega-y}
\end{align}

Finally, $\mathbf{C}(\mathbf{x}_{i})$ can be extracted from (\ref{eq:expectation-Omega-y})
as
\begin{align}
\mathbf{C}(\mathbf{x}_{i})= & \frac{M(\bar{N}_{\mathrm{t}}\mathbb{E}\{\hat{\boldsymbol{\Omega}}_{\text{y}}(\mathbf{x}_{i})\}-(N_{\mathrm{t}}-M)\text{tr}(\mathbb{E}\{\hat{\boldsymbol{\Omega}}_{\text{y}}(\mathbf{x}_{i})\})\mathbf{I})}{N_{\mathrm{t}}(N_{\mathrm{t}}M+N_{\mathrm{t}}-2)}\nonumber \\
 & -\frac{\sigma_{\text{n}}^{2}(M\bar{N}_{\mathrm{t}}\boldsymbol{\Omega}_{\mathrm{A}}(\mathbf{x}_{i})-N_{\mathrm{t}}(N_{\mathrm{t}}-M)\text{tr}(\boldsymbol{\Omega}_{\mathrm{A}}(\mathbf{x}_{i}))\mathbf{I})}{M(N_{\mathrm{t}}M+N_{\mathrm{t}}-2)}\label{eq:actual-sample-covariance}
\end{align}
where $\bar{N}_{\mathrm{t}}\triangleq N_{\mathrm{t}}^{2}+N_{\mathrm{t}}-2$
and $\boldsymbol{\Omega}_{\mathrm{A}}(\mathbf{x}_{i})\triangleq\frac{1}{|\mathcal{T}_{i}|}\sum_{t\in\mathcal{T}_{i}}\mathbf{A}_{t}^{\mathrm{H}}\mathbf{A}_{t}$.

Since trace is a linear operator, we immediately see that our estimator
$\hat{\mathbf{C}}(\mathbf{x}_{i})$ in (\ref{eq:unbiased-estimate-sampling-covariance})
is unbiased as $\mathbb{E}\{\hat{\mathbf{C}}(\mathbf{x}_{i})\}=\mathbf{C}(\mathbf{x}_{i})$.



\end{document}